\documentclass[11pt, letterpaper]{amsart}
\usepackage{graphicx, amssymb}

\newtheorem{theorem}{Theorem}[section]

\newtheorem{corollary}[theorem]{Corollary}

\newtheorem{definition}[theorem]{Definition}

\newtheorem{lemma}[theorem]{Lemma}

\newtheorem{proposition}[theorem]{Proposition}
\newtheorem{remark}[theorem]{Remark}

\numberwithin{theorem}{section}

\newcommand{\mc}{\mathcal}
\newcommand{\ra}{\rightarrow}
\newcommand{\mb}{\mathbb}
\DeclareMathOperator{\es}{ess\,sup \ }
\DeclareMathOperator{\ei}{ess\,inf \ }
\newcommand{\conv}{\text{conv}}

\numberwithin{equation}{section}

\title[]{On the Stability of Utility Maximization Problems}
\author[]{Erhan Bayraktar}\thanks{The authors are supported in part by the National Science Foundation under an applied mathematics research grant and a Career grant, DMS-0906257 and DMS-0955463, respectively, in part by the Susan M. Smith Professorship, and in part by the NDSEG Fellowship Program of the Department of Defense.} 
\address[Erhan Bayraktar]{Department of Mathematics, University of Michigan, 530 Church Street, Ann Arbor, MI 48104, USA}
\email{erhan@umich.edu}
\author[]{Ross Kravitz}
\address[Ross Kravitz]{Department of Mathematics, University of Michigan, 530 Church Street, Ann Arbor, MI 48104, USA}
\email{ross.kravitz@gmail.com}

\date{March 24, 2011}

\begin{document}
\begin{abstract}
In this paper we extend the stability results of \cite{MR2438002}. Our utility maximization problem is defined as an essential supremum of conditional expectations of the terminal values of wealth processes, conditioned on the filtration at the stopping time $\tau$.  To establish these results, our principal contribution is an extension of the classical result of convex analysis that pointwise convergence of convex functions implies convergence of their derivatives.  The notion of convex compactness introduced in \cite{Zit09c} plays an important role in our analysis. 
\end{abstract}

\maketitle
\tableofcontents

\section{Introduction}\label{sec:introduction}
In this paper, we extend the results of \cite{MR2438002} on the stability of the utility maximization problem with respect to changes in the market paramaters.  The main difference between our paper and theirs is that we work at arbitrary stopping times instead of only at the initial time $t = 0$.  We give a direct extension of \cite{MR2438002}, by specifiying when convergence of wealths and markets at a stopping time $\tau$ give rise to convergent optimal terminal wealths and convergent value functions.  Our results can also be used to treat the time zero case of \cite{MR2438002} when the initial sigma algebra $\mc{F}_0$ is nontrivial.  Along the way, we prove conditional versions of the convex duality of \cite{MR1722287}; parts of this theory have already been used, for example in \cite{MR2014244}, to prove time zero results.

In \cite{MR2438002}, a basic methodology is established for proving continuity of utility maximization problems.  First, one proves stability for the dual value problem, which is an optimization problem over the set of supermartingale deflators, the polar set of admissible wealth processes.  Second, to show that there is a ``continuous" connection between the dual and primal problems, one shows that the derivative of the dual value function is also stable with respect to perturbations of the market.  

In the time zero case, the second part of this program is trivial.  Indeed, it is a classical result from convex analysis that pointwise convergence of convex functions implies locally uniform convergence, which in turn yields convergence of derivatives.  Compactness plays a crucial role in establishing this theorem: essentially one uses an equicontinuity result and the existence of finite $\epsilon$-nets for compact sets.

Working in a conditional framework, we are led unavoidably to mappings from $L_{++}^0$ to $L^0$, where the topological structure is much less friendly.  Due to the scarcity of compact sets in the infinite dimensional, non locally convex space $L^0$, we are forced to work with the weaker concept of convexly compact sets, recently defined in \cite{Zit09c}.  In this setting, we use two kinds of generalized $\epsilon$-nets, one each for upper and lower bounds, to establish uniform convergence of dual value functions on convexly compact sets.  Here, we see that convexity alone is not enough, and we must use some additional structure of the dual value problem.  

As a corollary of our results, we can extract information about convergence of optimal wealth processes at intermediate times by exploiting a natural martingale property of optimal wealth processes.  For example, suppose that we put ourselves in the exact framework of \cite{MR2438002}.  We have a sequence of positive initial wealths $x_n$ converging to $x$, and a sequence of markets $\lambda_n$ converging in an appropriate sense to $\lambda$.  In \cite{MR2438002}, it is established that the optimal terminal wealths in each market, $\hat{X}^n_T$, converge in probability to $\hat{X}_T$, the optimal terminal wealth in the market $\lambda$.  Using our results, we may establish the fact that for any stopping time $\tau$, the optimal wealths at time $\tau$, $\hat{X}^n_\tau,$ converge to $\hat{X}_\tau$ in probability.

The rest of the paper is organized as follows: In the rest of Section~\ref{sec:introduction},
 we introduce the necessary financial framework for the problem, as well as highlighting the results of \cite{MR2438002} and \cite{MR1722287} on the stability problem and value function duality in general.  Finally, we state our main results in the paper.  In Section~\ref{sec:condual}, we extend results of convex analysis from the real-valued case to functions from $L^0$ to $L^0$, and establish the convex duality of \cite{MR1722287} in this setting, before applying our abstract results to the financial model.  In Section~\ref{sec:dc}, we prove that the dual value function is continuous with respect to dual wealth and market parameters.  In Section~\ref{sec:contode}, we show that the derivatives of the value functions, suitably defined, are also continuous with respect to the market parameters.  Finally, in Section~\ref{sec:fdtp}, we move from the dual problem to the primal one, and finish the proofs of our main theorems. In the appendix we establish a conditional version of the minimax theorem, which is used in Section $2$.


\subsection{The Financial Framework}

Let $(\Omega,\mc{F},P,(\mc{F}_t)_{t \in [0,T]})$ be a filtered probability space satisfying the usual conditions.  We assume that $\mc{F}_T = \mc{F}$, the global sigma algebra; if no specific sigma algebra is specified, $L^0$, $L^0_+$, $L^1$, etc. will always refer to measurability with respect to $\mc{F}$.  All random variables under consideration are at the least measurable with respect to $\mc{F}$.  Statements concerning random variables are always understood to hold almost surely.  Let $M$ be a continuous local martingale, and let \[\Lambda = \left \{ \lambda \ : \ \lambda \text{ is a predictable process satisfying } \int_0^T \lambda_u^2 d[M]_u < \infty \ \right \}.\]

For $\lambda \in \Lambda$, define 
\[S_t^\lambda = 1 + M_t + \int_0^t \lambda_u d[M]_u.\]  
Along with a numeraire bond, each $S^\lambda$ defines a stock market.  It is interpreted as the discounted price of an asset.  Let $\Lambda_m \subset \Lambda$ contain those $\lambda$ which define a market satisfying no free lunch with vanishing risk (NFLVR).  According to the paper \cite{MR1304434}, the NFLVR condition is equivalent to the existence of a local martingale measure for $S^\lambda$.  Also, it is proven in \cite{MR1384360} that all continuous market models satisfying NFLVR have the specific form described above.

A trading strategy $H$ is a predictable, $S$-integrable process.  We denote by $\mc{X}^\lambda(x)$ the set of wealth processes attainable from initial capital $x$ and subject to an admissibility constraint.  Formally,
\[ \mc{X}^\lambda(x) = \{x + H \cdot S^\lambda \ : \ H \text{ is predictable, } S-\text{integrable, and } x + H \cdot S^\lambda \geq 0 \}.\]
We will simply write $\mc{X}^\lambda$ for $\mc{X}^\lambda(1)$.

 For a given $\lambda \in \Lambda$, define $$Z_t^\lambda = \mathcal{E}(-\lambda \cdot M)_t = \exp \left ( -\int_0^t \lambda_u dM_u - \frac{1}{2} \int_0^t \lambda_u^2d[M]_u \right ).$$  This is a strictly positive local martingale such that $Z^\lambda X$ is a supermartingale for $X \in \mc{X}^\lambda(x)$.  Let $\mc{Y}^\lambda(y)$ be the set of supermartingale deflators starting from $y$ for the market described by $\lambda$.  Formally,
\[\mc{Y}^\lambda(y) = \{Y \ : \ Y \text{ is c\`{a}dl\`{a}g, adapted, positive, and }  XY \text{ is a supermartingale for all } X \in \mc{X}^\lambda, Y_0 = y\}.\]

We denote by $\overline{Y}^\lambda$ the set $\{Y_T \ : \ Y \in \mc{Y}^\lambda\}$ of terminal values of supermartingale deflators.  For $\tau$ a stopping time, we denote by $\mc{Y}^\lambda_\tau$ the set of supermartingale deflators starting at $1$ for the market restricted to the (random) time interval $[\tau,T]$.  Again, $\overline{\mc{Y}}^\lambda_\tau$ refers to the terminal values of elements of $\mc{Y}^\lambda_\tau$. Furthermore, for $\eta \in L^0_{++}(\mc{F}_\tau)$, i.e. $\eta$ strictly positive and $\mc{F}_\tau$-measurable, $\mc{Y}^\lambda_\tau(\eta)$ are those supermartingale deflators for $[\tau,T]$ which start at $\eta$.  Obviously, supermartingales starting at nonintegrable $\eta$ need not have finite expectation, so without further ado we relax this condition to require only that conditional expectations at time $\tau$ be finite.  We note that $\mc{Y}^\lambda(y) = y \mc{Y}^\lambda$ for $y > 0$.  More generally, for $\eta \in L^0_{++}(\mc{F}_\tau)$, we have $\mc{Y}^\lambda_\tau(\eta) = \eta \mc{Y}^\lambda_\tau$.  On the primal side, we define $\overline{\mc{X}}^\lambda$, $\overline{\mc{X}}_\tau^\lambda$, etc. accordingly.  

It is useful to consider the dual problem because the structure of $\mc{Y}^\lambda$ as a function of $\lambda$ is easy to understand.  From Proposition $3.2$ of \cite{MR2438002}, we have that all $Y \in \mc{Y}^\lambda$ such that $Y_T > 0$ have the form $Y = Z^\lambda \mc{E}(L) D$, where $L$ is a c\`{a}dl\`{a}g local martingale strongly orthogonal to $M$, and $D$ is a predictable, nonincreasing, c\`{a}dl\`{a}g process with $D_0 = 1$ and $D_T > 0$.  The extension of this result to $\mc{Y}^\lambda_\tau$ is trivial, with the obvious small modifications.  Furthermore, instead of considering $\mc{E}(L)$ as above, it is equivalent to simply consider a strictly positive c\`{a}dl\`{a}g local martingale $L'$ which is strongly orthogonal to $M$.

In light of the above description of $\mc{Y}^\lambda$, we will sometimes find it useful to write $\mc{Y}^\lambda = \{Z^\lambda \mc{Y}D \ : \ D \text{ as above }\}$, where $\mc{Y}$ is defined to be the set of strictly positive, c\`{a}dl\`{a}g local martingales which are strongly orthogonal to $M$, and starting from $1$.  The sets $\mc{Y}_t$, $\overline{\mc{Y}}$, and $\overline{\mc{Y}}_t$ are defined in the same way as before.  

We will consider utility functions defined on the positive axis.  We assume that $U:\mb{R}_+ \ra \mb{R}$ is $\mc{C}^1$, strictly concave, and satisfies the Inada conditions.  Most importantly, $U$ must have asymptotic elasticity strictly less than $1$, as defined in \cite{MR1722287}.  This means
\[AE(U) = \underset{x \ra \infty}{\limsup} \ \frac{xU'(x)}{U(x)} < 1,\]
and reflects the economic fact that the ratio of marginal utility to average utility should asymptotically become small.  The classical primal utility maximization problem is given by
\[ u_0^\lambda(x) = \underset{\{X_T : X \in \mc{X}^\lambda(x)\}}{\sup} \ E [U(X_T)].\]
In the following sections, we will be concerned with an extension of this utility maximization from $t=0$ to $[0,T]$-valued stopping times.

The convex conjugate of $U$ is denoted by $V$, and is defined, for $y>0$, by
\[V(y) = \underset{x > 0}{\sup} \ (U(x) - xy).\]   From basic facts of convex analysis, we know that $V$ is $\mc{C}^1$ and strictly convex.  Frequently, it will be convenient to decompose $V = V^+ - V^-$ into its positive and negative parts.  As a result of the asymptotic elasticity hypothesis on $U$, $V$ has the following property: there is $y_0 > 0$ such that for any $\mu \in (0,1)$, $V(\mu y) < \mu^{-\alpha} V(y)$ for any $y \in (0,y_0]$ and for some $\alpha > 0$; see Lemma 6.3 of \cite{MR1722287}.  The classical dual utility minimization problem is given by
\[v_0^\lambda(y) = \underset{Y \in Y^\lambda(y)}{\inf} E [V(Y_T)].\]

The fundamental result concerning the value functions $u_0^\lambda$ and $v_0^\lambda$ is provided by \cite{MR1722287}.  We cite the pertinent results here, suppressing $\lambda$ notation, because the results hold for a fixed market.

\begin{proposition}[\cite{MR1722287}]\label{KStheorem}  Assume that NFLVR is satisfied, the Inada conditions on $U$ hold, that $u_0(x) < \infty$ for some $x > 0$, and that the asymptotic elasticity of $U$ is less than $1$.  Then
\begin{enumerate}
\item [(a)] $u_0(x) < \infty \text{ for all } x,\ v_0(y) < \infty \text{ for all } y$.
\item [(b)] The value functions $u_0$ and $v_0$ are convex conjugates.
\item [(c)]  The value functions $u_0$ and $v_0$ are continuously differentiable on $(0,\infty)$.
\item [(d)] The optimal solution $\widehat{Y}(y) \in \mc{Y}(y)$ to the dual optimization problem exists, and the optimal solution $\widehat{X}(x) \in \mc{X}(x)$ to the primal optimization problem exists.  For $y = u_0'(x)$, we have the dual relation $\widehat{Y}(y)_T = U'(\widehat{X}(x)_T)$.
\end{enumerate}

\end{proposition}

\subsection{Introduction to the Stability Problem}

The following crucial definition is introduced in \cite{MR2438002}.

\begin{definition} A set $\Lambda'$ is \textbf{V-relatively compact} if $\left \{ V(Z_T^\lambda) \ : \ \lambda \in \Lambda' \right \}$ is uniformly integrable.
\end{definition}

\begin{definition}A topology $\mathfrak{T}$ on $\Lambda$ is called \textbf{appropriate} if the mapping $(\Lambda, \mathfrak{T}) \rightarrow L^0(\mathcal{F}_T)$ given by $\lambda \mapsto Z_T^\lambda$ is continuous with respect to convergence in probability.
\end{definition}

We have the following result from \cite{MR2438002}:

\begin{proposition}\label{LZmain}  Let $\Lambda'$ be a $V$-relatively compact subset of $\Lambda_m$, and let $\mathfrak{T}$ be an appropriate topology on $\Lambda$.  Then the mappings $$\Lambda' \times (0,\infty) \ni (\lambda, x) \mapsto u_0^\lambda(x) \in \mb{R}$$ and $$\Lambda' \times (0,\infty) \ni (\lambda, x) \mapsto \widehat{X}_T^{x,\lambda} \in L_{++}^0$$ are both jointly continuous.
\end{proposition}

Throughout this paper, we will frequently posit the existence of a sequence $(\lambda_n) \subset \Lambda'$ which is converging appropriately to some $\lambda \in \Lambda'$.  For any $\lambda$, we will denote by $\widehat{X}_T^{x,\lambda}$ the optimal terminal wealth in the agent's utility maximization problem.  $\widehat{Y}_T^{y,\lambda}$ will denote the terminal value of the dual minimization problem for $\lambda$.  For conciseness, $\widehat{X}^{x,\lambda_n}_T$ will be shortened to $\widehat{X}^{x,n}_T$, and $\widehat{Y}_T^{y,\lambda_n}$ will be shortened to $\widehat{Y}_T^{y,n}$.

To finish this section, we collect some known results from \cite{MR2438002} as well as some easy consequences.  We work under the same assumptions made in Proposition \ref{KStheorem}.

\begin{lemma}\label{vui}  The set $V^-(\overline{\mc{Y}}^\lambda)$ is uniformly integrable.
\end{lemma}

\begin{proof}  See \cite{MR1722287} for details.  The set $\overline{\mc{Y}}^\lambda$ is bounded in $L^1$, and $V^-$ is strictly concave.  The assertion follows now from the de La Vall\'{e}e-Poussin criterion for uniform integrability.
\end{proof}

A trivial consequence of Lemma \ref{vui} is that, for $y_0 > 0$, the set $\{V^-(y\overline{\mc{Y}}^\lambda) \ : \ y \in [0,y_0] \}$ is also uniformly integrable.

\begin{lemma}\label{uminus}  For $x>0$, the set $\left \{ U^-(\widehat{X}_T^{x,\lambda}) \ : \ \lambda \in \Lambda' \right \}$ is uniformly integrable.
\end{lemma}

\begin{proof}  Write the duality relationship $U(\widehat{X}_T^{x, \lambda}) = V(\widehat{Y}_T^{y,\lambda}) + \widehat{X}_T^{x,\lambda} \widehat{Y}_T^{y,\lambda}$, for $y = (u_0^\lambda)'(x)$.  The second term on the right hand side is nonnegative, so it follows that $0 \leq U^-(\widehat{X}_T^{x,\lambda}) \leq V^-(Y_T^{y,\lambda})$.  Now apply Lemma \ref{vui} to obtain the result.
\end{proof}

Again, we have a trivial extension of Lemma~\ref{uminus} to cases where $x$ is allowed to vary in an interval bounded away from zero as opposed to being held fixed.  Note that if $\mathcal{I}$ is a subinterval of $\mb{R}_{++}$ bounded away from zero, then $u_0'(\mathcal{I})$ is bounded from above.

\begin{lemma}\label{terminalconv}  Let $\lambda_n \rightarrow \lambda$ appropriately, and let $x_n \ra x$ in $\mb{R}_+$.  Then $U(\widehat{X}_T^{x_n,n}) \rightarrow U(\widehat{X}_T^{x,\lambda})$ in $L^1$.
\end{lemma}

\begin{proof} From Proposition \ref{LZmain}, we have that $\widehat{X}_T^{x_n,n} \ra \widehat{X}_T^{x,\lambda}$ in probability and that $E [U(\widehat{X}_T^{x_n,n})] \ra E [U(\widehat{X}_T^{x,\lambda})]$.  Since $U$ is continuous, so is $U^-$, and the first fact implies that $U^-(\widehat{X}_T^{x_n,n}) \rightarrow U^-(\widehat{X}_T^{x,\lambda})$ in probability.  Recall that convergence in probability plus uniform integrability is equivalent to $L^1$ convergence.  Then Lemma \ref{uminus} implies that $U^-(\widehat{X}_T^{x_n,n}) \rightarrow U^-(\widehat{X}_T^{x,\lambda})$ in $L^1$.

It remains to treat the positive part.  Note that $E [U(\widehat{X}_T^{x_n,n})] \rightarrow E [U(\widehat{X}_T^{x,\lambda})]$ and $U^-(\widehat{X}_T^{x_n,n}) \rightarrow U^-(\widehat{X}_T^{x,\lambda})$ in $L^1$ imply that $E [U^+(\widehat{X}_T^{x_n,n})] \rightarrow E [U^+(\widehat{X}_T^{x,\lambda})]$.  Furthermore, we of course have that $U^+(\widehat{X}_T^{x_n,n}) \rightarrow U^+(\widehat{X}_T^{x,\lambda})$ in probability.  Note that in Scheffe's lemma, a.s. convergence can just as easily be replaced by convergence in probability.  Thus, we have that $U^+(\widehat{X}_T^{x_n,n}) \rightarrow U^+(\widehat{X}_T^{x,\lambda})$ in $L^1$.  Putting the positive and negative pieces together, we obtain the result.
\end{proof}

\subsection{Main Financial Results}

The principal aim of this paper is to extend the results of \cite{MR2438002}, valid at $t = 0$, to all stopping times valued in $[0,T]$.  Let $\tau$ be a stopping time, and let $\xi \in L^0_{++}(\mc{F}_\tau)$ and $\eta \in L^0_{++}(\mc{F}_\tau)$.  We set
\[ 
u^\lambda_\tau(\xi) \triangleq \underset{X \in \overline{\mc{X}}^\lambda_\tau}{\es} E [ U(\xi X) \ | \ \mc{F}_\tau ], \quad 
 v^\lambda_\tau(\eta) \triangleq \underset{Y \in \overline{\mc{Y}}^\lambda_\tau}{\ei} E [ V(\eta Y) \ | \ \mc{F}_\tau].\]

\begin{theorem}\label{mainthm2}  Let $\Lambda'$ be a $V$-relatively compact subset of $\Lambda_m$, with an appropriate topology on $\Lambda$, and let $\tau$ be a $[0,T]$-valued stopping time.  Then the mappings 
\[\Lambda' \times L^0_{++}(\mc{F}_\tau) \ni (\lambda, \xi) \mapsto u_\tau^\lambda(\xi) \in L^0(\mc{F}_\tau)\] and 
\[\Lambda' \times L^0_{++}(\mc{F}_\tau) \ni (\lambda, \xi) \mapsto \widehat{X}_T^{\xi,\lambda} \in L_{++}^0\] 
are both jointly continuous.
\end{theorem}

\begin{theorem}\label{mainthm1}  Let $\tau$ be a $[0,T]$-valued stopping time.  Assume that NFLVR is satisfied for a given $\lambda \in \Lambda$, the Inada conditions on $U$ hold, that $u_\tau^\lambda(\xi)<\infty$ for some $\xi \in L^0_{++}(\mc{F}_\tau)$, and that the asymptotic elasticity of $U$ is less than $1$.  Then

\begin{enumerate}
\item $u_\tau^\lambda(\xi)<\infty$ for all $\xi \in L^0_{++}(\mc{F}_\tau)$ and $v_\tau^\lambda(\eta)<\infty$ for all $\eta \in L^0_{++}(\mc{F}_\tau)$.
\item $u_\tau^\lambda$ and $v_\tau^\lambda$ are both differentiable (as maps from $L^0 \ra L^0$), and their derivatives vary continuously when all spaces are endowed with the $L^0$ topology.
\item $v^\lambda_\tau(\eta) = \underset{\xi \in L^0_{++}(\mc{F}_\tau)}{\es} (u^\lambda_\tau(\xi) - \xi \eta)$, and $u^\lambda_\tau(\xi) = \underset{\eta \in L^0_{++}(\mc{F}_\tau)}{\ei} (v^\lambda_\tau(\eta) + \eta \xi)$, i.e. $u^\lambda_\tau$ and $v^\lambda_\tau$ are conjugate.
\item The optimal solution $\hat{Y}^{\eta,\lambda} \in \mc{Y}^\lambda_\tau$ to the conditional dual optimization problem exists, and the optimal solution $\hat{X}^{\xi,\lambda} \in \mc{X}^\lambda_\tau$ to the conditional primal optimization problem exists.  For $\eta = (u_\tau^\lambda)'(\xi) \in L^0_{++}(\mc{F}_\tau)$, we have the dual relation
$\eta \hat{Y}_T^{\eta,\lambda} = U'(\xi \hat{X}^{\xi,\lambda})$.
\end{enumerate}
\end{theorem}

Section~\ref{sec:condual} is devoted to the proof of this duality result, parts of which appear in \cite{MR2014244}.

We use the first part of Theorem \ref{mainthm2} to prove the corollary below.

\begin{corollary}\label{intcont}  Let $\Lambda'$ be a $V$-relatively compact subset of $\Lambda_m$, with an appropriate topology on $\Lambda$, and let $\tau$ be a stopping time.  Then the mapping
\[\Lambda' \times \mathbb{R}_{++} \ni (\lambda, x) \mapsto \hat{X}^{x,\lambda}_\tau \in L^0_{++}(\mc{F}_\tau)\]
is continuous. 
\end{corollary}

\section{Convex Duality}\label{sec:condual}

Let $\mc{G}$ be an arbitrary sub-sigma algebra of $\mc{F}$, which in applications will have the form $\mc{F}_\tau$, for $\tau$ a stopping time.  In this section, we prove a duality relationship between abstract versions of $u^\lambda_\tau$ and $v^\lambda_\tau$, employing the Minimax Theorem of Appendix A.  Afterwards, we show that this abstract framework encompasses the particular case we are interested in.  This section is based on Sections $3$ and $4$ of \cite{MR1722287}, as well as parts of \cite{MR1451876}.

We state here $\mc{G}$-measurable analogs of convexity for sets and functions.  The former concept has been defined in, i.e. \cite{MR1883202}.

\begin{definition} A set $K \subset L^0$ will be called $\mc{G}$-convex if, for $x,y \in K$ and $g \in m\mc{G}$ such that $0 \leq g \leq 1$, $gx + (1-g)y \in K$ also.
\end{definition}

\begin{definition} Let $f:L^0 \ra L^0$.  We say that $f$ is $\mc{G}$-convex if for any $\mc{G}$-measurable random variable $g$ such that $0 \leq g \leq 1$, and $x_1,x_2 \in L^0$, we have $f(gx_1 + (1-g)x_2) \leq g f(x_1) + (1-g) f(x_2)$.  We say that $f$ is strictly $\mc{G}$-convex if the inequality above is strict on some nonneglible set provided that $g$ takes values other than $1$ and $0$.
\end{definition}

The definition of a $\mc{G}$-concave and strictly $\mc{G}$-concave function are defined as above, with the inequalities reversed.

\subsection{Some properties of $L^0 \ra L^0$ maps}

We collect here some mathematical results which generalize classical results from convex analysis.

\begin{lemma}\label{bidual}[Biduality]  Suppose that $\tilde{u}:L^0_{++}(\mc{G}) \ra L^0(\mc{G})$ is $\mc{G}$-concave, and that $\tilde{v}$ satisfies $\tilde{v}(\eta) = \underset{\xi \in L^0_{++}(\mc{G})}{\es} (\tilde{u}(\xi) - \xi \eta)$, for $\eta \in L^0_{++}(\mc{G})$.  Then 
\[\tilde{u}(\xi) = \underset{\eta \in L^0_{++}(\mc{G})}{\ei} (\tilde{v}(\eta) + \xi \eta).\]
\end{lemma}

\begin{proof}
Consider the set $F^*$ of pairs $(\eta,\mu) \in L^0_{++}(\mc{G}) \times L^0(\mc{G})$ such that the affine function $h(\xi) = \xi \eta + \mu$ majorizes $\tilde{u}$.  We have $h(\xi) \geq \tilde{u}(\xi)$ if and only if $\mu \geq \underset{\xi \in L^0_{++}(\mc{G})}{\es} \left ( \tilde{u}(\xi) - \eta \xi \right ) = \tilde{v}(\eta)$.  Thus, $F^*$ is seen to be the epigraph of $\tilde{v}$.  Since $\tilde{u}$ is $\mc{G}$-concave, it is the pointwise essential infimum of the affine functions $h(\xi) = \xi \eta + \mu$, for $(\eta,\mu) \in F^*$.  Thus, $\tilde{u}(\xi) = \underset{\eta \in L^0_{++}(\mc{G})}{\ei} \left (\tilde{v}(\eta) + \eta \xi \right )$, and both parts of the conjugacy relationship are established.
\end{proof}

For $z \in L^0_{++}(\mc{G})$, we say that $z^* \in -L^0_{++}(\mc{G})$ is a subdifferential of the $\mc{G}$-convex function $\tilde{v}$ at $z$ if $\tilde{v}(\eta) \geq \tilde{v}(z) + z^*(\eta - z)$ for all $\eta \in L^0_{++}(\mc{G})$, and we denote this by $z^* \in \partial\tilde{v}(z)$.
The superdifferential of $\tilde{u}$ is defined analagously, with the inequality above reversed.  As in the classical real-valued case, the bidual conjugacy between $\tilde{u}$ and $\tilde{v}$ and some algebraic manipulation implies that $z \in \partial\tilde{v}(z^*)$ if and only if $z^* \in \partial (-\tilde{u})(z)$.  

\begin{definition} We say that a $\mc{G}$-convex function $\tilde{v}$ is \textbf{differentiable} if its subdifferential contains a single element at each point in its domain.  \end{definition}

\begin{lemma}\label{strictdiff}  Suppose that $\tilde{v}$ is strictly $\mc{G}$-convex.  Then its conjugate $\tilde{u}$ is differentiable.
\end{lemma}

\begin{proof}  According to the discussion above, the superdifferential of $\tilde{u}$ is $-(\partial \tilde{v})^{-1}$, and this mapping is single-valued if and only if $\tilde{u}$ is differentiable.   Consequently, it suffices to show that $\partial \tilde{v}(\eta_1) \cap \partial \tilde{v}(\eta_2) = \emptyset$ for $\eta_1 \neq \eta_2$ in $L^0_{++}(\mc{G})$.  Suppose that $\eta^* \in \partial \tilde{v}(\eta_1) \cap \partial \tilde{v}(\eta_2)$.  The graph of $h(z) = \eta^*z - \tilde{u}(\eta^*)$ is a supporting hyperplane $H$ to $\text{epi } \tilde{v}$ that contains $(\eta_1,\tilde{v}(\eta_1))$ and $(\eta_2,\tilde{v}(\eta_2))$.  Hence, $\tilde{v}$ cannot be strictly $\mc{G}$-convex along the line segment joining $\eta_1$ and $\eta_2$.  

\end{proof}

\subsection{The Abstract Convex Duality Problem}\label{sec:absv}

Suppose that $\mc{C}$ and $\mc{D}$ are subsets of $L^0_+(\mc{F})$ which are 
\begin{enumerate}

\item $\mc{G}$-convex, solid, and closed in the topology of convergence in probability, and
\item $g \in \mc{C}$ if and only if $E[gh \ | \ \mc{G} ] \leq 1$ for all $h \in \mc{D}$, and $h \in \mc{D}$ if and only if $E [gh \ | \ \mc{G} ] \leq 1$ for all $ g \in \mc{C}$, and
\item The constant function $1$ is in $\mc{C}$.

\end{enumerate}

Note that the second and third conditions imply that $\mc{D}$ is contained in the unit ball of $L^1$.  For $\xi \in L^0_{+}(\mc{G})$, let $\mc{C}(\xi) = \xi\mc{C}$.  Define $\mc{D}(\eta)$ similarly for $\eta \in L^0_{++}(\mc{G})$.

We consider the abstract utility maximization problems 
\[\tilde{u}(\xi) \triangleq \underset{g \in \mc{C}(\xi)}{\es} E [ U(g) \ | \ \mc{G} ], \quad \tilde{v}(\eta) \triangleq \underset{h \in \mc{D}(\eta)}{\ei} E [V(h) \ | \ \mc{G}],\]
for $\xi \in L^0_{+}(\mc{G})$ and $\eta \in L^0_{++}(\mc{G})$.

Throughout this section, assume that $\tilde{u}(\xi) < \infty$ almost surely, for some $\xi$.  We also suppose that $\tilde{v}$ and $\tilde{u}$ satisfy the property of being \textbf{locally defined}.  More precisely, we say that a map $\tilde{v}$ is locally defined if, for $\eta_1,\eta_2 \in L^0_{++}(\mc{G})$ and $A \in \mc{G}$, $\tilde{v}(1_A \eta_1 + 1_{A^c} \eta_2) = 1_A \tilde{v}(\eta_1) + 1_{A^c} \tilde{v}(\eta_2)$.  The property is defined analagously for $\tilde{u}$ with respect to its domain.

 We state here the abstract version of Theorem \ref{mainthm1}

\begin{theorem}\label{abstractthm1}  Assume that the Inada conditions on $U$ hold, that $\tilde{u}(\xi)<\infty$ for some $\xi \in L^0_{++}(\mc{G})$, that the asymptotic elasticity of $U$ is less than $1$, and that $\tilde{u}$,$\tilde{v}$ are locally defined.  Then

\begin{enumerate}
\item $\tilde{u}(\xi)<\infty$ for all $\xi \in L^0_{++}(\mc{G})$ and $\tilde{v}(\eta)<\infty$ for all $\eta \in L^0_{++}(\mc{G})$.
\item $\tilde{u}$ and $\tilde{v}$ are both differentiable, and their derivatives vary continuously when all spaces are endowed with the $L^0$ topology.
\item $\tilde{v}(\eta) = \underset{\xi \in L^0_{++}(\mc{G})}{\es} \tilde{u}(\xi) - \xi \eta$, and $\tilde{u}(\xi) = \underset{\eta \in L^0_{++}(\mc{G})}{\ei} \tilde{v}(\eta) + \eta \xi$, i.e. $\tilde{u}$ and $\tilde{v}$ are conjugate
\item The optimal solution $\hat{Y}(\eta)$ to the conditional dual optimization problem exists, and the optimal solution $\hat{X}(\xi)$ to the conditional primal optimization problem exists.  For $\eta = \tilde{u}'(\xi) \in L^0_{++}(\mc{G})$, we have the dual relation
$\hat{Y}(\eta)_T = U'(\hat{X}(\xi)_T)$.
\end{enumerate}
\end{theorem}  

The rest of this subsection is devoted to the proof of Theorem~\ref{abstractthm1}. We first develop a few auxiliary results.

\begin{lemma}[Komlos's Lemma]\label{komlos} Let $(f^n)_{n \geq 1}$ be a sequence of non-negative random variables.  Then there is a sequence $g^n \in \text{conv}(f^n,f^{n+1},\ldots)$ which converges almost surely to a variable $g$ with values in $[0,\infty]$.  If the convex hull of $(f^n)_{n \geq 1}$ is bounded in probability, the limit $f$ is real-valued.
\end{lemma}

\begin{proof}  See Lemma A1.1 in \cite{MR1304434}.
\end{proof}

We state here a conditional version of uniform integrability.  Some simple properties related to this concept are proven in Section $3.1$.  In this section, they are used only for the following lemma, so we defer their proofs for ease of reading.

\begin{definition}  Let $\{X_\alpha \}_{\alpha \in A}$ be a collection of random variables.  We say that the collection is $\mc{F}_\tau$-\textbf{uniformly integrable} if for any $\epsilon(\omega) \in L^0_{++}(\mc{F}_\tau)$, there exists some $\delta(\omega) \in L^0_{++}(\mc{F}_\tau)$ such that for $B_\alpha = \{|X_\alpha| \geq \delta \}$, then $|E [ 1_{B_\alpha} X_\alpha \ | \ \mc{F}_\tau ]| < \epsilon$ for all $\alpha$.
\end{definition}

\begin{lemma}\label{ui}  For any $\eta \in L^0_{++}(\mc{G})$, the family $(V^-(h))_{h \in \mc{D}(\eta)}$ is $\mc{G}$-uniformly integrable.  If $(h^n)_{n \geq 1}$ is a sequence in $\mc{D}(\eta)$ which converges almost surely to a random variable $h$, then $h \in \mc{D}(\eta)$ and \[\underset{n \ra \infty}{\lim \inf} \ E [V(h^n) \ | \ \mc{G} ] \geq E [ V(h) \ | \ \mc{G} ]. \]
\end{lemma}

\begin{proof}  See Lemma 3.4 in \cite{MR1722287} and Section $3.1$.  The first claim is proved exactly as in that lemma, using the modified de La Vall\'{e}e-Poussin criterion for $\mc{G}$-uniform integrability (Lemma \ref{convp}), that $V^-$ is strictly concave, and that the conditional expectations of elements in $\mc{D}(\eta)$ are bounded in $L^0$.  The second claim is also proved as in \cite{MR1722287} using Lemma \ref{confatou}; it is merely the conditional version of the unconditional result given in \cite{MR1722287}.  
\end{proof}

\begin{lemma}\label{dex}  Suppose that $\tilde{v}(\eta) < \infty$ for $\eta \in L^0_{++}(\mc{G})$.  Then the optimal solution $\widehat{h}(\eta)$ to the dual optimization problem exists and is unique.  As a consequence, $\tilde{v}$ is strictly $\mc{G}$-convex on $\{ \tilde{v} < \infty \}$.
\end{lemma}

\begin{proof}  First, we claim that the set  $\left \{ E [V(g) \ | \ \mc{G} ] \right \}_{g \in \mc{D}(\eta)}$ is downwards directed.  Let $g_1,g_2 \in \mc{D}(\eta)$.  Let $A \triangleq \{E [V(g_1) \ | \ \mc{G} ] \leq E [ V(g_2) \ | \ \mc{G} ] \} \in \mc{G}.$  Let $g = 1_Ag_1 + 1_{A^c}g_2$.  Since $\mc{D}(\eta)$ is $\mc{G}$-convex, $g \in \mc{D}(\eta)$.  We calculate that 
\[
\begin{split}
E [V(g) \ | \ \mc{G}]
&= E [ 1_AV(g_1) + 1_{A^c}V(g_2) \ | \ \mc{G}]
\\&= 1_AE [ V(g_1) \ | \ \mc{G} ] + 1_{A^c}E [ V(g_2) \ | \ \mc{G} ]
\\&= E [ V(g_1) \ | \ \mc{G} ] \wedge E [ V(g_2) \ | \ \mc{G} ].
\end{split}
\]  

Since the above set is downwards directed, there exists a sequence $(g^n)_{n \geq 1}$ in $\mc{D}(\eta)$ such that $E [ V(g^n) \ | \ \mc{G} ] \downarrow \tilde{v}(\eta)$ a.s.  By Lemma \ref{komlos}, there exists a sequence $h^n \in \text{conv}(g^n,g^{n+1},\ldots)$ and a finite random variable $\widehat{h}$ such that $h^n \ra \widehat{h}$ a. s.  From the convexity of the function $V$ we have $E [ V(h^n) \ | \ \mc{G} ] \leq \underset{m \geq n}{\es} E [ V(g^m) \ | \ \mc{G} ]$, so that $\underset{n \ra \infty}{\lim} E [ V(h^n) \ | \ \mc{G} ] = \tilde{v}(\eta)$ a.s.

By Lemma \ref{ui}, $ E [ V(\widehat{h}) \ | \ \mc{G} ] \leq \underset{n \ra \infty}{\liminf} \  E [ V(h^n) \ | \ \mc{G} ] = \tilde{v}(\eta)$, and $\widehat{h} \in \mc{D}(\eta)$.  The uniqueness of the optimal solution follows from the strict convexity of $V$.  For the strict $\mc{G}$-convexity of $\tilde{v}$ on its effective domain, let $\eta_1,\eta_2 \in L^0_{++}(\mc{G})$, and let $g$ be between $0$ and $1$ and $\mc{G}$-measurable.  Then $g \widehat{h}(\eta_1) + (1-g)\widehat{h}(\eta_2)$ is an element of $\mc{D}(g \eta_1 + (1-g) \eta_2)$.  We have
\[ \tilde{v}(g\eta_1 + (1-g)\eta_2) \leq E [ V(g\widehat{h}(\eta_1) + (1-g)\widehat{h}(\eta_2)) \ | \ \mc{G} ].\]
Using the strict convexity of $V$, we have $V(g \widehat{h}(\eta_1) + (1-g)\widehat{h}(\eta_2)) \leq -\epsilon 1_A + g V(\widehat{h}(\eta_1)) + (1-g)V(\widehat{h}(\eta_2))$ provided that $g$ is not only $1$ and $0$, $A \in \mc{G}$ corresponding to some set of positive measure on which $\widehat{h}(\eta_1)$ and $\widehat{h}(\eta_2)$ are bounded from above and $g$ is bounded away from $0$ and $1$.  The strict $\mc{G}$-convexity of $\tilde{v}$ is now immediate.
\end{proof}

\begin{lemma}\label{duality}  We have $\tilde{v}(\eta) = \underset{\xi \in L^0_{++}(\mc{G})}{\es} (\tilde{u}(\xi) - \xi\eta)$ for each $\eta \in L^0_{++}(\mc{G})$.
\end{lemma}

\begin{proof}  For $n > 0$, define $\mc{B}_n$ to be the the positive elements of the ball of radius $n$ in $L^\infty(\mc{G})$.  The sets $\mc{B}_n$ are $\sigma(L^\infty,L^1)$-compact.  By the conditional Minimax Theorem, we have, for $n$ fixed and all $\eta \in L^0_{++}(\mc{G})$:  \[\underset{g \in \mc{B}_n}{\es} \underset{h \in \mc{D}(\eta)}{\ei} E [ U(g) - gh \ | \ \mc{G} ] = \underset{h \in \mc{D}(\eta)}{\ei} \underset{g \in \mc{B}_n}{\es} E [ U(g) - gh \ | \ \mc{G} ];\] 
we use this fact later.  From the dual relation between the sets $\mc{C}(\xi)$ and $\mc{D}(\eta)$, we deduce that $g \in \mc{C}(\xi)$ if and only if 
\[\underset{h \in \mc{D}(\eta)}{\es} E [gh \ | \ \mc{G}] \leq \xi\eta.\]  We claim that the following quantities are equal:

\begin{enumerate}

\item $\underset{n \ra \infty}{\lim} \underset{g \in B_n}{\es} \underset{h \in \mc{D}(\eta)}{\ei} E [ U(g) - gh \ | \ \mc{G} ]$
\item $\underset{\xi \in L^0_{+}(\mc{G})}{\es} \underset{ g \in \mc{C}(\xi)}{\es} E [ U(g) - \xi \eta \ | \ \mc{G} ]$
\item $\underset{\xi \in L^0_{++}(\mc{G})}{\es} \underset{ g \in \mc{C}(\xi)}{\es} E [ U(g) - \xi \eta \ | \ \mc{G} ]$

\end{enumerate}

\begin{itemize}

\item ``$(1) \leq (2)$'':  It suffices to prove that for any $n$, and any $g \in \mc{B}_n$, there exists some $\overline{\xi} \in L^0_+(\mc{G})$ such that 
\[\underset{h \in \mc{D}(\eta)}{\ei} E [U(g) - gh \ | \ \mc{G}] \leq \underset{g \in \mc{C}(\overline{\xi})}{\es} E [ U(g) - \overline{\xi}{\eta} \ | \ \mc{G} ].\]  
Take $\overline{\xi} \in L^0_+(\mc{G})$ that is minimal with respect to $g$ being in  $\mc{C}(\overline{\xi})$; such a $\overline{\xi}$ satisfies, by the duality relationship, $\underset{h \in \mc{D}(\eta)}{\es} E [gh \ | \ \mc{G} ] = \overline{\xi}\eta$.  In fact this duality shows that such a $\overline{\xi}$ exists, because we can take 
\[\overline{\xi} = \frac{1}{\eta}\underset{h \in \mc{D}(\eta)}{\ei} E [ gh \ | \ \mc{G}].\]  
Then $\underset{h \in \mc{D}(\eta)}{\ei} E [ U(g) - gh \ | \ \mc{G}] = E [ U(g) - \overline{\xi}\eta \ | \ \mc{G}] \leq \underset{g \in \mc{C}(\overline{\xi})}{\es} E [ U(g) - \overline{\xi}\eta \ | \ \mc{G} ].$

\item ``$(2) \leq (3)$'':  It suffices to show that for any $\xi \in L^0_+(\mc{G})$ and any $g \in \mc{C}(\xi)$, there exists a $\overline{\xi} \in L^0_{++}(\mc{G})$ and a $\overline{g} \in \mc{C}(\overline{\xi})$ such that $E [ U(g) - \xi \eta \ | \ \mc{G} ] \leq E [ U(\overline{g}) - \overline{\xi}\eta \ | \ \mc{G} ]$.  Note that $g = 0$ on $\{\xi = 0 \}$, by definition of $\mc{C}(\xi)$.  Let $A = \{\xi = 0 \} \in \mc{G}$.  For each natural number $k$, write $B_k = \{k-1 \leq \eta < k\} \in \mc{G}$.  Obviously the $B_k$'s partition $\Omega$, so write $A = \underset{k}{\bigcup} (A \bigcap B_k) \triangleq \underset{k}{\bigcup} A_k$.  Recalling that $U$ satisfies the Inada condition, for each $k$, choose $\epsilon_k \downarrow 0$ such that $U'(\epsilon_k) > k$.  Define 
\[\overline{\xi} = \xi + \sum_{k=1}^\infty \epsilon_k 1_{A_k},\] 
so that $\overline{\xi} > 0$.  Let $\overline{g} = g_{1_{A^c}} + \overline{\xi}_{1_A}$.  Claim that $\overline{g} \in \mc{C}(\overline{\xi})$.  First, $\xi \leq \overline{\xi}$ implies, by the duality relationship, that $\mc{C}(\xi) \subset \mc{C}(\overline{\xi})$, so that $g \in \mc{C}(\overline{\xi})$.  Noting that $A \in \mc{G}$, use the $\mc{G}$-convexity of $\mc{C}(\overline{\xi})$ and $\overline{\xi} \in \mc{C}(\overline{\xi})$ to conclude that $\overline{g} \in \mc{C}(\overline{\xi})$.

Now, we calculate, by the convexity of $U$, that 
\[U(\overline{g}) - \overline{\xi} \eta \geq U(g) - \xi \eta  + \sum_{k=1}^\infty (U'(\epsilon_k)\epsilon_k - \epsilon_k \eta)1_{A_k}.\]  Note that every term in the summand on the right hand side is $\mc{G}$-measurable, and that each summand is also positive by construction.  Hence, we take conditional expectations and deduce that $E [ U(\overline{g}) - \overline{\xi}\eta \ | \ \mc{G} ] \geq E [ U(g) - \xi \eta \ | \ \mc{G} ]$.

\item ``$(3) \leq (1)$'':  First, as usual, note that for fixed $\xi$, $\{E [ U(g) - \xi \eta \ | \ \mc{G}]\}_{g \in \mc{C}(\xi)}$ is upwards directed by the $\mc{G}$-convexity of $\mc{C}(\xi)$.  We claim that also $\left \{ \underset{g \in \mc{C}(\xi)}{\es} E [ U(g) - \xi \eta \ | \ \mc{G} ] \right \}_{\xi \in L^0_{++}(\mc{G})}$ is ``almost'' upwards directed, in a sense to be described below.  Fix $\xi^1$ and $\xi^2$ in $L^0_{++}(\mc{G})$.  Take $g^1_n \in \mc{C}(\xi^1)$ and $g^2_n \in \mc{C}(\xi^2)$ such that $f^i(g^i_n) \triangleq E [ U(g^i_n) - \xi^i \eta \ | \ \mc{G}] \uparrow \underset{g \in \mc{C}(\xi^i)}{\es} E [ U(g) - \xi^i \eta \ | \ \mc{G} ] \triangleq F(\xi_i)$ for each $i$.  Let $A = \left \{\underset{g \in \mc{C}(\xi^1)}{\es} E [ U(g) - \xi^1 \eta \ | \ \mc{G} ] \geq \underset{g \in \mc{C}(\xi^2)}{\es} E [ U(g) - \xi^2 \eta \ | \ \mc{G} ] \right \} \in \mc{G}$.  Let $\xi = 1_A\xi^1 + 1_{A^c}\xi^2$.  Let $g_n = 1_Ag^1_n + 1_{A^c} g^2_n$.  It follows by the duality relationship that $g_n \in \mc{C}(\xi)$ for each $n$.  Furthermore, 
\[ f(g_n) \triangleq E [ U(g_n) - \xi \eta \ | \ \mc{G} ] = 1_A E [U(g^1_n) - \xi^1 \eta \ | \ \mc{G} ] + 1_{A^c} E [ U(g^2_n) - \xi^2  \eta \ | \ \mc{G} ],\]
and this quantity converges upwards towards $F^1(\xi^1) \vee F^(\xi^2)$.  From this it follows that $F(\xi) \geq F(\xi^1) \vee F(\xi^2)$.  This isn't quite upwards directedness, but it is sufficient for the essential supremum to be realized by an increasing sequence.

Now, let $ g \in \mc{C}(\overline{\xi})$ for some arbitrary but fixed $\overline{\xi}$.  Let $g_n = g1_{\{g \leq n\}} \in \mc{C}(\overline{\xi}) \cap B_n$ by the solidness of $\mc{C}(\overline{\xi})$.  Then by the duality relationship, 
\[\underset{h \in \mc{D}(\eta)}{\ei} E [U(g_n) - g_nh \ | \ \mc{G} ] \geq E [ U(g_n) - \overline{\xi}\eta \ | \ \mc{G} ].\]  
By conditional Monotone Convergence, we have $E [ U(g_n) - \overline{\xi}\eta \ | \ \mc{G} ] \uparrow E [ U(g) - \overline{\xi}\eta \ | \ \mc{G} ]$.  

Now we put the two above parts together.  Let $y_j \in \mc{C}(\xi_j)$ be such that $E [ U(y_j) - \xi_j \eta \ | \ \mc{G} ] \uparrow \underset{\xi \in L^0_{++}(\mc{G})}{\es} \underset{ g \in \mc{C}(\xi)}{\es} E [ U(g) - \xi \eta \ | \ \mc{G} ]$.  Define $y_j^n = y_j1_{\{y_j \leq n\}} \in \mc{B}_n \cap \mc{C}(\xi_j)$, so that $y_j^n \uparrow y_j$ for all $j$.  This gives $E[U(y_j^n) - \overline{\xi} \eta \ | \ \mc{G}] \uparrow E [ U(y_j) - \overline{\xi} \eta \ | \ \mc{G}]$.  Using an appropriate diagonal, we have the existence of a sequence $(n_k)_{k \geq 1}$ such that 
\[\lim_{k \ra \infty} y_k^{n_k} = \underset{\xi \in L^0_{++}(\mc{G})}{\es} \underset{ g \in \mc{C}(\xi)}{\es} E [ U(g) - \xi \eta \ | \ \mc{G} ],\] 
and the original claim is proven.
\end{itemize}

So, we have 
\[
\begin{split}
\lim_{n \ra \infty} \underset{ g \in B_n}{\es} \underset{h \in \mc{D}(\eta)}{\ei} E [ U(g) - gh \ | \ \mc{G} ] &= \underset{\xi \in L^0_{++}(\mc{G})}{\es} \underset{g \in \mc{C}(\xi)}{\es} E [ U(g) - \xi \eta \ | \ \mc{G} ] \\&
 = \underset{\xi \in L^0_{++}(\mc{G})}{\es} (\tilde{u}(\xi) - \xi \eta).
 \end{split}
 \]

On the other hand, 
\[\underset{h \in \mc{D}(\eta)}{\ei} \underset{g \in B_n}{\es} E [ U(g) - gh \ | \ \mc{G} ] = \underset{h \in \mc{D}(\eta)}{\ei} E [V^n(h) \ | \ \mc{G}] \triangleq \tilde{v}^n(\eta),\] where $V^n(y) = \underset{0 \leq x \leq n}{\sup} [U(x) - xy ]$.  Consequently, in light of the minimax result of Appendix A, it is enough to show that 
\[\lim_{n \ra \infty} \tilde{v}^n(\eta) = \underset{n \ra \infty}{\lim} \  \underset{h \in \mc{D}(\eta)}{\ei} E [ V^n(h) \ | \ \mc{G} ] = \tilde{v}(\eta).\]  
Clearly, $\tilde{v}^n \leq \tilde{v}$ because $V^n \leq V$.  As before, the $\mc{G}$-convexity of $\mc{D}(\eta)$ implies that for all $n$, $\{E [ V^n(h) \ | \ \mc{G} ] \ : \ h \in \mc{D}(\eta) \}$ is downward directed.  Thus, by diagonalization, let $(h^n)_{n \geq 1}$ be a sequence in $\mc{D}(\eta)$ such that 
$\underset{n \ra \infty}{\lim} E [ V^n(h^n) \ | \ \mc{G} ] = \lim_{n \ra \infty} \tilde{v}^n(\eta)$.  
By Lemma \ref{komlos}, there exists a sequence $f^n \in \conv(h^n,h^{n+1},\ldots)$ which converges almost surely to a variable $h$.  We have $h \in \mc{D}(\eta)$ because $\mc{D}(\eta)$ is closed under convergence in probability.  Since $V^n(y) = V(y)$ for $ y \geq I(1) \geq I(n)$, where $I(\cdot)$ is the negative inverse of $V'(\cdot)$, we know from Lemma \ref{ui} that $(V^n(f^n)^-)_{n \geq 1}$ is uniformly integrable.  Thus, as is proven before, the convexity of $V^n$ and conditional Fatou's lemma imply that 
\[\lim_{n \ra \infty} E [V^n(h^n) \ | \ \mc{G} ] \geq \liminf_{n \ra \infty} E [V^n(f^n) \ | \ \mc{G} ] \geq E [ V(h) \ | \ \mc{G} ] \geq \tilde{v}(\eta).\]  
This shows that $\lim_{n \ra \infty} \tilde{v}^n(\eta) = \tilde{v}(\eta)$.
\end{proof}

\begin{lemma}\label{derivcont}  Let $\eta_k$ be elements of $L^0_{++}(\mc{G})$ converging in probability to $\eta \in L^0_{++}(\mc{G})$.  Let $\widehat{h}(\eta_k) = \arg \min \tilde{v}(\eta_k)$, and let $\widehat{h}(\eta) = \arg \min \tilde{v}(\eta)$, i.e. the optimal dual variables.  Then $\widehat{h}(\eta_k) \ra \widehat{h}(\eta)$ in probability.
\end{lemma}

\begin{proof}  See the first part of \cite{MR1722287}, Lemma $3.8$.  The proof here is essentially identical, and basically a consequence of the strict convexity of $V$.
\end{proof}

\begin{lemma}\label{dderivcont}  Assume the same hypotheses we have in Lemma~\ref{derivcont}.  Then $E [V'(\widehat{h}(\eta_k))\widehat{h}(\eta_k) \ | \ \mc{G}]$ converges to $E [V'(\widehat{h}(\eta))\widehat{h}(\eta) \ | \ \mc{G}]$ in probability.
\end{lemma}

\begin{proof} The proof is again identical to that of Lemma $3.9$ in \cite{MR1722287}.
\end{proof}

\begin{remark}\label{rem:adrm}  Suppose that $\mu_k$ is a sequence of $\mc{G}$-measurable random variables converging uniformly to $1$.  Then we can still conclude that $E\left[V'\left(\mu_k \widehat{h}(\eta_k)\right) \widehat{h}(\eta_k) \ | \ \mc{G}\right]$ converges to $E\left[V'\left(\widehat{h}(\eta)\right) \widehat{h}(\eta) \ | \ \mc{G}\right]$ in probability.  The reasoning is identical to that of \cite{MR1722287} in Remark $3.1$.
\end{remark}

We say that $\tilde{v}$ is G\^{a}teaux differentiable at $\eta \in L^0_{++}(\mc{G})$ in the direction $b \in L^0_{++}(\mc{G})$ if 
\[\tilde{v}'(\eta;b) \triangleq \lim_{s \ra 0} \frac{\tilde{v}(\eta + sb) - \tilde{v}(\eta)}{s}\] 
exists as a limit in probability.  We denote by $\tilde{v}^+(\eta;b)$ the one-sided right-hand G\^{a}teaux derivative, calculated only as $s \downarrow 0$.

We first assume that $\eta$ is bounded away from zero, and that $b \in L^\infty(\mc{G})$.  This ensures that for $s$ sufficiently small, $\eta + sb \in L^0_{++}(\mc{G})$, the domain of $\tilde{v}$.

\begin{lemma}\label{vdiffer}  The limit $\tilde{v}'(\eta;b)$ exists in probability for $\eta$ and $b$ as above, and is equal to $\frac{-b}{\eta}E[\hat{h}(\eta)I(\hat{h}(\eta)) \ | \ \mc{G}]$, where $I=-V'=(U')^{-1}$.
\end{lemma}

\begin{proof}  The proof is very similar to Lemma $3.10$ of \cite{MR1722287}, but differs slightly on the technical details.  First, we claim that $-\eta \tilde{v}'(\eta;b) = \underset{s \ra 1}{\lim} \frac{\tilde{v}(\eta) - \tilde{v}(s^b \eta)}{s - 1},$ provided that this limit exists in probability.  We have 
\[
\begin{split}
\underset{s \ra 1}{\lim} \frac{\tilde{v}(\eta) - \tilde{v}(s^b \eta)}{s - 1} &= \underset{s \ra 1}{\lim} \frac{\tilde{v}(\exp(\log \eta)) - \tilde{v}(\exp(\log s^b+ \log \eta))}{s - 1}\\&= \underset{s \ra 1}{\lim} \frac{\tilde{v}_e(\log \eta) - \tilde{v}_e(b\log s + \log \eta)}{s - 1},
\end{split}
\]
 where $\tilde{v}_e(\eta) = \tilde{v}(\exp(\eta))$.  As $s \ra 1$, $\log s = s + o(s)$.  Hence, this last quantity would be equal to $-\tilde{v}_e'(\log \eta; b)$.  We then calculate that $-\tilde{v}_e'(\log \eta;b) = -\tilde{v}'(\exp(\log \eta);b)\exp(\log \eta) = -\tilde{v}'(\eta;b)\eta$.

Note that all of the above was contingent on the limit being well-defined in probability.  We now show that 
\[P-\limsup_{s \downarrow 1} \frac{\tilde{v}(\eta) - \tilde{v}(s^b \eta)}{s - 1} \leq - b E [\widehat{h}(\eta)I(\widehat{h}(\eta)) \ | \ \mc{G}],\] and that 
\[P-\liminf_{s \downarrow 1} \frac{\tilde{v}(\eta) - \tilde{v}(s^b \eta)}{s - 1} \geq - b E [\widehat{h}(\eta)I(\widehat{h}(\eta)) \ | \ \mc{G}].\]  By Lemma \ref{derivcont}, the map $\eta \mapsto E [\widehat{h}(\eta)I(\widehat{h}(\eta)) \ | \ \mc{G}]$ is continuous in probability.  Thus, if we can prove the above two inequalities, we will have shown that $\tilde{v}^+(\eta;b)$ is continuous in $\eta$.  Given that $\tilde{v}$ is $\mc{G}$-convex, we know, as in the real-valued case, that this is sufficient to prove the differentiability of $\tilde{v}$ in the direction $b$.  If there is a point of non-differentiability, there must be a discontinuity of the right-handed derivative.

For the first inequality, we have 
\[
\begin{split}
P-\limsup_{s \downarrow 1} \frac{\tilde{v}^s(\eta) - \tilde{v}(s^b \eta)}{s - 1} &\leq P-\limsup_{s \downarrow 1} \frac{1}{s - 1} E [V(\frac{1}{s^b}\widehat{h}(s^b \eta)) - V(\widehat{h}(s^b \eta)) \ | \ \mc{G}]
\\ & \leq P-\limsup_{s \downarrow 1} \frac{1}{s - 1} E [(\frac{1}{s^b} - 1)\widehat{h}(s^b \eta)V'(\frac{1}{s^b}\widehat{h}(s^b \eta)) \ | \ \mc{G}]\\
& = P-\limsup_{s \downarrow 1} \frac{1}{s - 1}(\frac{1}{s^b} - 1) E [\widehat{h}(s^b \eta)V'(\frac{1}{s^b}\widehat{h}(s^b \eta)) \ | \ \mc{G}]\\
 &= -b E [\widehat{h}(\eta)I(\widehat{h}(\eta)) \ | \ \mc{G}],
 \end{split}
 \] 
by Remark~\ref{rem:adrm}, using the fact that $b \in L^\infty(\mc{G})$.

For the second inequality, we have 
\[
\begin{split}
P-\liminf_{s \downarrow 1} \frac{\tilde{v}(\eta) - \tilde{v}(s^b \eta)}{s - 1} &= P-\liminf_{s \downarrow 1} \frac{1}{s - 1} E [ V(\widehat{h}(\eta)) - V(s^b \widehat{h}(\eta)) \ | \ \mc{G}]\\&\geq \liminf_{s \downarrow 1} \frac{1}{s - 1}E [(1 - s^b)\widehat{h}(\eta)V'(s^u\widehat{h}(\eta)) \ | \mc{G}]
\\& = -b E [\widehat{h}(\eta) I(\widehat{h}(\eta)) \ | \ \mc{G}],
\end{split}\] 
where the last equality follows from the conditional monotone convergence theorem.

\end{proof}

We cannot at first define the G\^{a}teaux derivative of $\tilde{v}$ for all points in its domain, due to the fact that $L^0_{++}(\mc{G})$ is not open in $L^0(\mc{G})$ whenever the underlying set $\Omega$ is not finite.  Using the locally defined nature of $\tilde{v}$, we can still actually define $\tilde{v}(\eta;b)$ for arbitrary $\eta \in L^0_{++}(\mc{G})$ and $b \in L^1(\mc{G})$ as a kind of ``$\sigma$-derivative".  We extend the definition of G\^ateaux derivative to all $\eta \in L^0_{++}(\mc{G})$ and $b \in L^1(\mc{G})$ as follows.  

First, let $\eta \in L^0_{++}(\mc{G})$ and $b \in L^\infty_{++}(\mc{G})$.  For natural numbers $n$, let $A_n = \{\eta \geq \frac{1}{n}\}$.  Let $\eta_n = \eta1_{A_n} + \frac{1}{n}1_{(A_n)^c}$.  Note that $\eta_n$ is bounded from below, so that for $b \in L^\infty(\mc{G})$, the directional derivative $\tilde{v}'(\eta_n;b)$ is well-defined.  For $m > n$, note that $\eta_m$ and $\eta_n$ agree on $A_n$, so by the local property of $\tilde{v}$, we see that $\tilde{v}'(\eta_m;b)$ and $\tilde{v}'(\eta_n;b)$ agree on $A_n$.  Thus, we can safely define 
\[\tilde{v}'(\eta;b) = \sum_{n=1}^\infty \tilde{v}'(\eta_n;b)1_{A_n \setminus A_{n-1}}.\]

We now do a similar thing to extend the allowable $b$ values.  For $b \in L^1(\mc{G})$, let $B_n = \{|b|  \leq n\}$ for $n \geq 1$.  Let $b_n = b 1_{B_n} + n1_{(B_n)^c}$.  Since each $b_n$ is bounded, $\tilde{v}'(\eta;b_n)$ is well-defined.  Again using the local property of $\tilde{v}$, we consistently define, for $\eta \in L^0_{++}(\mc{G})$, $\tilde{v}'(\eta;b) = \tilde{v}'(\eta;b_n)$ on $B_n$. 

According to Lemma \ref{vdiffer}, for $\eta \in L^0_{++}(\mc{G})$ and bounded away from zero, we have $\tilde{v}'(\eta) = \frac{-1}{\eta}E[\hat{h}(\eta)I(\hat{h}(\eta)) \ | \ \mc{G}]$, noting the relationship between the subderivative and G\^{a}teaux derivative.   The mapping $\eta \mapsto E[\hat{h}(\eta)I(\hat{h}(\eta)) \ | \ \mc{G}]$ is defined locally, because $\eta \mapsto \hat{h}(\eta)$ is.  This implies that Lemma \ref{vdiffer} can be extended to all $\eta \in L^0_{++}(\mc{G})$.  We have:

\begin{lemma}\label{vdiffer2}  Let $\eta \in L^0_{++}(\mc{G})$.  Then $\tilde{v}'(\eta) = \frac{-1}{\eta}E [ \hat{h}(\eta) I(\hat{h}(\eta)) \ | \ \mc{G}]$.   
\end{lemma}

\begin{lemma}\label{pd}  Suppose that the random variables $\xi$ and $\eta$ are related by $\xi = -\tilde{v}'(\eta)$.  If $\hat{h}(\eta)$ is the optimal dual variable, then $\widehat{g}(\xi) \triangleq I(\widehat{h}(\eta))$ is the optimal primal variable.
\end{lemma}

\begin{proof}  See Lemma $3.11$ of \cite{MR1722287}.

\end{proof}

We are now ready to prove the main result of Section~\ref{sec:absv}.

\noindent \textbf{Proof of Theorem \ref{abstractthm1}}  By hypothesis, $\tilde{u}(\xi)<\infty$ for some $\xi$.  The concavity of $U$ implies that $\tilde{u}$ is $\mc{G}$-concave, and it follows that $\tilde{u}(\xi)<\infty$ for all $\xi$.  Lemma \ref{vdiffer} establishes that $\tilde{v}(\eta)<\infty$ for all $\eta$.  This finishes Part $1$.  

The continuous differentiability of $\tilde{v}$ follows from Lemmas \ref{vdiffer2} and \ref{dderivcont}.  The differentiability of $\tilde{u}$ follows from the strict $\mc{G}$-convexity of $\tilde{v}$, established in Lemma \ref{dex}, and Lemma \ref{strictdiff}.  By plugging in $\hat{g}(\xi) = I(\hat{h}(\eta))$ and $\eta = \tilde{u}(\xi)$ into the formula for $\tilde{v}'$, we obtain $\xi \tilde{u}'(\xi) = E [ \hat{g}(\xi) U'(\hat{g}(\xi)) \ | \ \mc{G} ]$, which is shown without too much difficulty to be continuous.  This is Part $2$.

Part $3$ follows from Lemmas \ref{duality} and \ref{bidual}, and Part $4$ is the content of Lemma \ref{pd}
\hfill $\square$

\subsection{Convex Duality in the Financial Setting}

We prove Theorem \ref{mainthm1} here.  First, we establish the local property of $v_\tau^\lambda$, since in Section~\ref{sec:absv} we assumed this property to hold.

\begin{lemma}  Let $A \in \mc{F}_\tau$, and let $\eta_1,\eta_2 \in L^0_{++}(\mc{F}_\tau)$.  Then $v_\tau^\lambda(1_A\eta_1 + 1_{A^c}\eta_2) = 1_Av_\tau^\lambda(\eta_1) + 1_{A^c}v_\tau^\lambda(\eta_2)$.
\end{lemma}

\begin{proof}  Write $\eta = 1_A\eta_1 + 1_{A^c}\eta_2$.  By the $\mc{F}_\tau$-convexity of $v_\tau^\lambda$, we know that $v_\tau^\lambda(\eta) \leq 1_Av_\tau^\lambda(\eta_1) + 1_{A^c}v_\tau^\lambda(\eta_2)$.  Now, we calculate 
\[
\begin{split}
v_\tau^\lambda(\eta) &= \underset{Y \in \mc{Y}_\tau}{\ei} E [ V(Z_T^\lambda \eta Y_T) \ | \ \mc{F}_\tau]
\\&= \underset{Y \in \mc{Y}_\tau}{\ei} E [ 1_AV(Z_T^\lambda \eta_1 Y_T) + 1_{A^c} V(Z_T^\lambda \eta_2 Y_T) \ | \ \mc{F}_\tau]\\
& = \underset{Y \in \mc{Y}_\tau}{\ei} \left [ 1_A E [V(Z_T^\lambda \eta_1 Y_T) \ | \ \mc{F}_\tau ] + 1_{A^c} E [ V(Z_T^\lambda \eta_2 Y_T) \ | \ \mc{F}_\tau ] \right ]
\\& \geq \underset{Y \in \mc{Y}_\tau}{\ei}  1_A E [V(Z_T^\lambda \eta_1 Y_T) \ | \ \mc{F}_\tau ] +  \underset{Y \in \mc{Y}_\tau}{\ei} 1_{A^c} E [ V(Z_T^\lambda \eta_2 Y_T) \ | \ \mc{F}_\tau ]
\\& = 1_Av_\tau^\lambda(\eta_1) + 1_{A^c} v_\tau^\lambda(\eta_2).
\end{split}
\]
\end{proof}
We now may prove the theorem.

\noindent \textbf{Proof of Theorem \ref{mainthm1}}
We apply the abstract results of the previous setting.  For $\tau$ a stopping time, we let $\mc{C}(\xi)$  be the solid hull of $\overline{\mc{X}}^\lambda_\tau(\xi)$, with $\xi \in L^0_{++}(\mc{F}_\tau)$, and we let $\mc{D}(\eta)$ be the solid hull of $\overline{\mc{Y}}^\lambda_\tau(\eta)$ for $\eta \in L^0_{++}(\mc{F}_\tau)$.  Note that by passing to solid hulls above does not change either the primal nor the dual value function, because $U$ is increasing and $V$ is decreasing. 

We must prove that $\mc{C}(\xi)$ and $\mc{D}(\eta)$ satisfies the properties of the previous subsection.  \v{Z}itkovi\'{c} proves this exact result in Theorem $4$ of \cite{MR1883202}. Thus, the application of Theorem \ref{abstractthm1} completes the proof.
\hfill $\square$

\section{Dual Continuity}\label{sec:dc}

In this section, we prove that for any stopping time $\tau$, the mapping $(\eta,\lambda) \mapsto v^\lambda_\tau(\eta)$ is continuous from $L^0_{++}(\mc{F}_\tau) \times \Lambda'$ to $L^0(\mc{F}_\tau)$.  All spaces of random variables above are topologized by convergence in probability, and the source space has the product topology.  The next two subsections contain some technical results needed to work with random variables which have finite $\mc{F}_\tau$-conditional expectation but are not necessarily integrable.  We state the lemmas, but leave the proofs to the reader, as they are similar to their classical counterparts.

\subsection{Mathematical Preliminaries for Proving Continuity}

\begin{lemma}{(Conditional Dominated Convergence)}\label{cdcon}  Suppose that $Y_k$ are positive random variables converging in probability to the positive random variable $Y$.  Suppose that we have $Y_k \leq X_k \triangleq \eta_k U_k$ for each $k$, $X_k \ra X \triangleq \eta U$ in probability, and $E [ X_k \ | \ \mc{F}_\tau] \ra E [ X \ | \ \mc{F}_\tau ]$ in probability, where $\eta_k,\eta \in L^0_{++}(\mc{F}_\tau)$ and $U_k,U \in L^1_{++}$.  (In particular, all the quantities given are well-defined and finite).  Then $E [ Y_k \ | \ \mc{F}_\tau] \ra E [Y \ | \ \mc{F}_\tau ]$ in probability.
\end{lemma}

\begin{proof}  Recall that convergence in probability is equivalent to the fact that from any subsequence, one can extract a subsubsequence converging almost surely.  Take an arbitrary subsequence $\left \{ E [Y_{k_n} \ | \ \mc{F}_\tau ] \right \}_{n \geq 1}$.  We must find a subsubsequence that converges almost surely to $ E [ Y \ | \ \mc{F}_\tau ]$.  Extract a subsubsequence $(k_{n_j})$ such that $Y_{k_{n_j}} \ra Y$ almost surely and $X_{k_{n_j}} \ra X$ almost surely.  Hence, it suffices to prove this result in the case that every instance of convergence in probability is replaced with almost sure convergence.

Note that both $X_k + Y_k$ and $X_k - Y_k$ are nonnegative random variables for all $k$.  By the generalized conditional Fatou's Lemma, and the hypothesis, we have $$ E [ X + Y | \ \mc{F}_\tau ] \leq \liminf_{k \ra \infty} E [ X_k + Y_k \ | \ \mc{F}_\tau ] = E [ X \ | \ \mc{F}_\tau] + \liminf_{k \ra \infty} E [ Y_k \ | \ \mc{F}_\tau ],$$ $$E [ X - Y \ | \ \mc{F}_\tau ]  \leq \liminf_{k \ra \infty} E [ X_k - Y_k \ | \ \mc{F}_\tau ] = E [ X \ | \ \mc{F}_\tau] - \liminf_{k \ra \infty} E [ Y_k \ | \ \mc{F}_\tau ].$$  Subtracting $E [ X \ | \ \mc{F}_\tau] < \infty$ from both sides, we obtain $$ E [ Y \ | \ \mc{F}_\tau] \leq \liminf_{k \ra \infty} E [ Y_k \ | \ \mc{F}_\tau ] \text{ and } E [ Y \ | \ \mc{F}_\tau ] \geq \limsup_{k \ra \infty} E [ Y_k \ | \ \mc{F}_\tau ],$$ and this implies the result.
\end{proof}

Given a sequence of random variables $(X_n)_{n \geq 1}$, we define the \textbf{limit superior in probability} of the $X_n$, denoted $P-\limsup$, by 
\[P-\limsup_{n \ra \infty} X_n \triangleq \ei \left \{X \ : \ \text{for all } \epsilon > 0, \ P(X < X_n - \epsilon) \ra 0 \right \}.\]  
We define the \textbf{limit inferior in probability} of the $X_n$, denoted $P-\liminf$, by \[P-\liminf_{n \ra \infty} X_n \triangleq \es \left \{X \ : \ \text{for all } \epsilon > 0, P(X < X_n + \epsilon) \ra 0 \right \}.\]  
Note that $X \geq P- \limsup X_n$ and $X \leq  P-\liminf X_n$ is equivalent to $X_n \ra X$ in probability.  Furthermore, since almost sure convergence is stronger than convergence in probability, it is always true that $P - \underset{n \ra \infty}{\limsup} \  X_n \leq \underset{n \ra \infty}{\limsup} \  X_n$ and $P - \underset{n \ra \infty}{\liminf} \  X_n \geq \underset{n \ra \infty}{\liminf} \  X_n$.  It is clear that each of these inequalities may be strict.  Before proceding, we clarify the question of existence of $P-\limsup$'s and $P-\liminf$'s.  By definition, such random variables must be unique.  We prove existence here for $P-\limsup$'s, the argument for $P-\liminf$'s being identical.

\begin{lemma}  Let $X_n$ be a sequence of random variables.  Then $P-\limsup X_n$ exists in the sense of an extended random variable. 
\end{lemma}

\begin{proof}  Consider the set of extended random variables $\mc{A} = \{A \ : \ P(A - X_n < -\epsilon) \ra 0 \text{ for all } \epsilon > 0\}$.  Obviously $\infty \in \mc{A}$, so that $\mc{A} \neq \emptyset$.  We claim that the set $\mc{A}$ is downwards directed.  Indeed, for $A_1,A_2 \in \mc{A}$ and $\epsilon > 0$, $P(A_1 \wedge A_2 - X_n < - \epsilon) \leq P(A_1 - X_n < - \epsilon) + P(A_2 - X_n < - \epsilon) \ra 0$.  Downward directedness implies that there exists a sequence $A_1,A_2,\ldots \in \mc{A}$ with the property that $A_n \downarrow A \triangleq \ei \mc{A}$.  We claim that $A$ is the $P-\limsup$ of the $X_n$.  It suffices to show that $A \in \mc{A}$.  Otherwise, there exists $\epsilon > 0$ such that for infinitely many $n$, $P(A - X_n < - \epsilon) > \epsilon$.  By Egorov's Theorem, there exists a suitably large $j$ such that $|A - A_j| < \frac{\epsilon}{2}$ except possibly on a set of measure less than $\frac{\epsilon}{2}$.  These two observations imply that for infinitely many $n$, $P(A_j - X_n < -\frac{\epsilon}{2}) > \frac{\epsilon}{2}$, but this contradicts the fact that $A_j \in \mc{A}$.  Hence, $A \in \mc{A}$.
\end{proof}

\begin{lemma}{(Conditional de la Vall\'{e}e-Poussin)}\label{convp}  Let $\{X_\alpha \}_{\alpha \in A}$ be some family of random variables.  Let $G(x)$ be a nonnegative increasing function such that $\underset{x \ra \infty}{\lim} \frac{G(x)}{x} = \infty$.  Then if $\{E [G(X_\alpha) \ | \ \mc{F}_\tau ] \}_{\alpha \in A}$ is bounded in probability, then $\{X_\alpha\}_{\alpha \in A}$ is $\mc{F}_\tau$-uniformly integrable.
\end{lemma}


\begin{proof}  Suppose for contradiction that the family $\{X_\alpha\}_{\alpha \in A}$ is not $\mc{F}_\tau$-uniformly integrable.  Then there exists an $\epsilon \in L^0_{++}(\mc{F}_\tau)$ such that for any $\delta \in L^0_{++}(\mc{F}_\tau)$, there exist $X_\alpha$ and  $B = B(\delta,\alpha) \in \mc{F}_\tau$ such that $E [1_B X_\alpha \ | \ \mc{F}_\tau ] \geq \epsilon$, where $B(\delta,\alpha) = \{X_\alpha \geq \delta\}$.  Let $(X_n,\delta_n)$ be pairs as above, for $\delta_n$ tending uniformly to  $\infty$ (we'll just set $\delta_n = n$).  Suppose then that $\frac{G(x)}{x} \geq k_n$ for $x \geq n$, where $k_n \uparrow \infty$.  Then, we have $$E [G(X_n) \ | \ \mc{F}_\tau ] \geq E [1_{B(\delta_n,n)} G(X_n) \ | \ \mc{F}_\tau ] \geq k_n E [1_{B(\delta_n,n)} X_n \ | \ \mc{F}_\tau ] \geq k_n \epsilon.$$  This clearly contradicts the bounded in probability hypothesis.
\end{proof}

\begin{lemma}\label{confatou}{(Conditional Fatou's Lemma)}  Let $(X_n)_{n \geq 1}$ be a sequence of random variables such that $X_n \ra X$ in probability.  Suppose that the negative parts $(X_n^-)_{n \geq 1}$ are $\mc{F}_\tau$-uniformly integrable.  Then a version of Fatou's Lemma holds: $$E \left [ \underset{n \ra \infty}{\liminf} X_n \ | \ \mc{F}_\tau \right ] \leq \underset{n \ra \infty}{P-\liminf} E [ X_n \ | \ \mc{F}_\tau].$$
\end{lemma}

\begin{proof}  The proof is standard.  Considering the negative parts, by using an increasing sequence of truncations, we apply Lemma \ref{cdcon} to obtain convergence in probability of the conditional expectations of the negative parts.  We then just apply the standard conditional Fatou's Lemma to the positive part in order to obtain the result.
\end{proof}

Suppose that in some metric space $\Theta$, we have $x_i \ra x$.  Given a mapping $\phi$ from $\Theta$ into $L^0$, we will say that $\phi$ is \textbf{upper semi-continuous} if $\phi(x) \geq P-\limsup \phi(x_i)$.  We say that $\phi$ is \textbf{lower semi-continuous} if $\phi(x) \leq P-\liminf \phi(x_i)$.

\begin{lemma}\label{infsc}  For some arbitrary index set $\mathfrak{I}$, let $f_i:\ \ra L^0$, for each $i$, be an upper semi-continuous mapping from an arbitrary metric space $\Theta$ into $L^0$.  Then the mapping $f$ from $\Theta$ to $L^0$ defined by $$f(x) = \underset{i \in \mathfrak{I}}{\ei} f_i(x)$$ is also upper semi-continuous.
\end{lemma}

\begin{proof}  By a familiar subsequence type argument, it suffices to prove the lemma replacing the probabilistic convergence with almost sure convergence.  Let $x_j \ra x$ in $X$.  We must show that $f(x) \geq \underset{x_j \ra x}{\limsup} f(x_j)$.  For each $i \in I$, we have by hypothesis that $f_i(x) \geq \underset{x_j \ra x}{\limsup} f_i(x_j)$.  Taking essential infimums over $i$ yields $f(x) \geq \underset{i \in I}{\ei} \underset{x_j \ra x}{\limsup} f_i(x_j) \geq \underset{x_j \ra x}{\limsup} \ \underset{i \in I}{\ei} f_i(x_j) = \underset{x_j \ra x}{\limsup} f(x_j)$, where the last inequality follows from the same reasoning as in the easy direction of the minimax theorem.
\end{proof}

\subsection{Proving Continuity}

\begin{lemma}\label{dualexistence}  Let $\lambda \in \Lambda'$ and suppose that $v_\tau^\lambda(\eta) < \infty$ for $\eta \in L^0_{++}(\mc{F}_\tau)$.  Then there exists a local martingale $L^{\eta,\lambda}$, strictly orthogonal to $M$ when restricted to $[\tau,T]$, such that $v_\tau^\lambda(\eta) = E [ V(\eta Z_T^\lambda \mc{E}(L^{\eta,\lambda})) \ | \ \mc{F}_\tau].$
\end{lemma}
\begin{proof}  By Lemma \ref{dex}, there exists $Y \in \mc{Y}^\lambda_\tau$ such that $v_\tau^\lambda(\eta) = E [ V(\eta Y_T) \ | \ \mc{F}_\tau ]$.  By Proposition $3.2$ of \cite{MR2438002}, $Y = Z^\lambda \mc{E}(L)D$ with $L$ strongly orthogonal to $M$ on $[\tau,T]$, and $D$ is a decreasing process.  Because $V$ is strictly decreasing, and because $Z^\lambda \mc{E}(L) \in \mc{Y}^\lambda_\tau$ by virtue of $\langle L,M \rangle = 0$, it must be that $D \equiv 1$.
\end{proof}

\begin{corollary}\label{below} 
Let $\mc{B}$ denote the set of all local martingales $L$, strongly orthogonal to $M$ when both are restricted to the time interval $[\tau,T]$, such that the terminal value $\mc{E}(L)_T$ is bounded from below by a positive constant. Let $\lambda \in \Lambda'$ and suppose that $E [ V^+(Z_T^\lambda) \ | \ \mc{F}_\tau ] < \infty$.  Then for each $\eta \in L^0_{++}(\mc{F}_\tau)$, we have the representation 
\[v_\tau^\lambda(\eta) = \underset{L \in \mc{B}}{\ei} E [ V(\eta Z_T^\lambda \mc{E}(L)_T) \ | \ \mc{F}_\tau ].\]
\end{corollary}

\begin{proof}  The proof is similar to that of Corollary 3.4 from \cite{MR2438002}.  First, we write down a solution to the unrestricted optimization problem.  The optimal terminal value may not be bounded away from zero, so it is then truncated from below, and it can be shown that these approximations yield converging approximates to the value function.

\end{proof}

We now apply Lemma \ref{cdcon} to obtain the equivalent of Lemma $3.5$ in \cite{MR2438002}.  

\begin{lemma}\label{usc1}  Let $\xi$ be a random variable, bounded from below by a strictly positive constant, such that $\underset{\lambda \in \Lambda'}{\sup} E [ Z_T^\lambda \xi ] < \infty.$  Then the mapping $(\eta, \lambda ) \overset{\phi}{\mapsto} E [ V(\eta Z_T^\lambda \xi) \ | \ \mc{F}_\tau ]$ is upper semi-continuous from $L^0_{++}(\mc{F}_\tau) \times \Lambda'$ into $L^0(\mc{F}_\tau)$.
\end{lemma}

\begin{proof}  Define the map $\phi^+$ analogously to $\phi$, except that $V$ is replaced with $V^+$.  To prove the upper semi-continuity of $\phi$, it suffices to prove the continuity of $\phi^+$, because we get upper semi-continuity of the negative part for free by conditional reverse Fatou's Lemma.  In any case, we will now just assume that $V$ is a nonnegative function, and that $V(y) = 0$ for sufficiently large $y$.  Let $(\eta_k, Z_T^{\lambda_k}) \ra (\eta, Z_T^\lambda)$.  We put ourselves in the context of Lemma \ref{cdcon}.  Let $Y_k = V(\eta_k Y_T^{\lambda_k} \xi)$.  By the asymptotic elasticity assumption on $U$, we have by Lemma $6.3$ of \cite{MR1722287} that $Y_k \leq C \eta_k^{-\alpha} V(Y_T^{\lambda_k} \xi) + D \triangleq X_k$, where $C,D > 0$ are constants, as is $\alpha >0$.  For completeness, actually $\alpha = \frac{\gamma}{1 - \gamma}$, for $\gamma$ the asymptotic elasticity of $U$.  Clearly, by the continuity of $V$, we have $Y_k \ra Y = V(\eta Z_T^\lambda \xi)$ and $X_k \ra X = \eta^{-\alpha}V(Z_T^\lambda \xi)$, both in probability.  Since all $\eta_k$ and $\eta$ are $\mc{F}_\tau$-measurable, then $$E [ X_k \ | \ \mc{F}_\tau ] = \eta_k^{-\alpha} E [V(Z_T^{\lambda_k} \xi) \ | \ \mc{F}_\tau ] \ra \eta^{-\alpha} E [ V(Z_T^\lambda \xi) \ | \ \mc{F}_\tau ] = E [X \ | \ \mc{F}_\tau ],$$ where the convergence above is a consequence of the uniform integrability established in Lemma 3.5 of \cite{MR2438002}.  The proof is now complete by Lemma \ref{cdcon}.
\end{proof}

\begin{lemma}\label{usc}  The function $(\eta,\lambda) \mapsto v_\tau^\lambda(\eta)$, mapping $L^0_{++}(\mc{F}_\tau) \times \Lambda'$ into $L^0(\mc{F}_\tau)$ is upper-semi continuous.
\end{lemma}

\begin{proof}  By Corollary \ref{below}, the dual value function $v_\tau^\lambda(\cdot)$ has the representation $v_\tau^\lambda(\eta) = \underset{Y}{\ei} E [ V(\eta Y_T Z_T^\lambda) \ | \ \mc{F}_\tau ],$ where the infimum is taken over $Y= \mc{E}(L)$ for $L \in \mc{B}$.  For such a random variable, by Lemma \ref{usc1}, the mapping $(\eta,\lambda) \mapsto E [V(\eta Y Z_T^\lambda) \ | \ \mc{F}_\tau ]$ is upper semi-continuous.  Therefore, by Lemma \ref{infsc}, the mapping $(\eta,\lambda) \mapsto v_\tau^\lambda(\eta)$ is upper semi-continuous, as an essential infimum of upper semi-continuous mappings.
\end{proof}

We will prove lower semi-continuity by approximating the value function from below.  We have this preparatory lemma.

\begin{lemma}\label{auxlsc}  Let $\lambda_k \ra \lambda$ appropriately, and let $\eta_k \ra \eta$ in $L^0_{++}(\mc{F}_\tau)$, and suppose that $\eta_k,\eta$ are all uniformly bounded from below.  Then $v_\tau^{\lambda_k}(\eta_k) \ra v_\tau^\lambda(\eta)$, i.e. $(\eta,\lambda) \mapsto v_\tau^\lambda(\eta)$ is continuous when restricted to $\eta$ bounded away from zero.
\end{lemma}

\begin{proof}  Let $(\lambda_k,\eta_k) \ra (\lambda,\eta)$.  Lemma \ref{usc} implies that the mapping is upper semi-continuous.  Suppose that $v_\tau^\lambda$ is not lower semi-continuous.  By passing to a subsequence, we assume that there exists $\beta > 0$ such that for all $k$, $P(v_\tau^\lambda(\eta) > v_\tau^{\lambda_k}(\eta_k) + \beta) > \beta$, and that $\eta_k \ra \eta$ and $Z_T^{\lambda_k} \ra Z_T^\lambda$ almost   surely.  Since $\eta_k \ra \eta$ a.s., Egorov's Theorem (p. 73, \cite{MR924157}) implies that for any $\epsilon > 0$, there exists a set $A \in \mc{F}_\tau$ with $P(A) < \epsilon$, such that $\eta_k 1_{A^c}, \eta 1_{A^c}$ are all uniformly bounded from above, for $k$ sufficiently large, with an upper bound of $M \in \mathbb{R}_+$; to see this, we first find a small set outside of which $\eta$ is bounded, and we then use Egorov's theorem to find another small set outside of which the convergence is uniform.  We will take $\epsilon$ to be small compared to $\beta$.

Set $\mathfrak{v}_\tau^\lambda(\eta) \triangleq E [v_\tau^\lambda(\eta) 1_{A^c}]$.  By construction, $\mathfrak{v}_\tau^\lambda(\eta)$ is finite: for the negative parts, by Lemma \ref{dex}, we have $v_\tau^{\lambda_k}(\eta_k) = E [ V(\eta_k \widehat{Y}_T^k) \ | \ \mc{F}_\tau ]$, where the optimizer $\hat{Y}^k_T = \widehat{Y}_T^{\eta_k,\lambda_k}$ can be written as $Z_T^{\lambda_k}\mc{E}(L^k)_T$, for some local martingale $L^k = L^{\eta_k,\lambda_k}$ that is strongly orthogonal to $M$ on $[\tau,T]$.  Then  note that $v_\tau^{\lambda_k}(\eta_k) 1_{A^c} = E [ V(\eta_k \hat{Y}^k_T) \ | \ \mc{F}_\tau] 1_{A^c} \geq E [V(M \hat{Y}^k_T) \ | \ \mc{F}_\tau] 1_{A^c}$ since $A \in \mc{F}_\tau$, and this term's negative part is integrable, since $M \hat{Y}^k_T \in L^1$ and $V^-$ is strictly concave.  In fact, the negative parts parts of $v^{\lambda_k}_\tau(\eta_k)1_{A^c}$ are also uniformly integrable, since the collection $\{M \hat{Y}^k_T\}_{k \geq 1}$ is bounded in $L^1$.  The positive parts are controlled by the fact that $\eta$ is bounded away from zero together with the asymptotic elasticity hypothesis and $V$-compactness assumption (in fact this logic implies that the positive parts of $\{v^{\lambda_k}_\tau(\eta_k)\}_{k \geq 1}$  are uniformly integrable).

By Lemma \ref{usc}, we already know that $(\lambda,\eta) \mapsto v^\lambda(\eta)$ is upper semi-continuous. This means that $(v_\tau^\lambda(\eta) - v_\tau^{\lambda_k}(\eta_k))^- \ra 0$ in probability.  Thus, $\left(1_{A^c}(v_\tau^\lambda(\eta) - v_\tau^{\lambda_k}(\eta_k))\right)^- \ra 0$ in $L^1$, since these random variables are uniformly integrable.  Writing $\mathfrak{v}_\tau^\lambda(\eta) - \mathfrak{v}_\tau^{\lambda_k}(\eta_k) = 
E \left[\left(1_{A^c}(v_\tau^\lambda(\eta) - v_\tau^{\lambda_k}(\eta_k))\right)^+ \right] - E \left[\left(1_{A^c}(v_\tau^\lambda(\eta) - v_\tau^{\lambda_k}(\eta_k))\right)^- \right]$, it therefore must be the case that $\mathfrak{v}_\tau^\lambda(\eta) \geq \mathfrak{v}_\tau^{\lambda_k}(\eta_k) + \beta(\beta - \epsilon)$ for $k$ large, because $\left(1_{A^c}(v_\tau^\lambda(\eta) - v_\tau^{\lambda_k}(\eta_k))\right)^+ \geq \beta$ on a set of measure at least $\beta - \epsilon$.

Note that the collection $\{Z_T^{\lambda_k}\mc{E}(L^k)_T \}_{k \geq 1}$ is bounded in $L^1$.  Hence, by Komlos's Lemma, there exists a random variable $h \in L^0_+(\mc{F}_\tau)$ and $\alpha^k_j$ such that 
\[h_k = \sum_{j=k}^{J(k)} \alpha_j^k \eta_j Z_T^{\lambda_j} \mc{E}(L^j)_T \ra h \text{ a.s.},\] with the $\alpha_j^k > 0$ such that $\sum_{j = k}^{J(k)} \alpha_j^k = 1$ for all $k$; here we have used the fact that $\eta_k \ra \eta$ a.s., which is preserved under convex combinations.  This then implies that $$f_k = \sum_{j = k}^{J(k)} \alpha_j^k \mc{E}(L^j)_T \ra \frac{h}{\eta Z_T^\lambda}, \text{ almost surely}.$$  The random variables $f_n$ are all in $\overline{\mc{Y}}_\tau^0 \triangleq \overline{\mc{Y}}_\tau^{\lambda \equiv 0}$ (since each $L^k$ is orthogonal to $M$), which is closed with respect to convergence in probability by Lemma $4.1$ of \cite{MR1722287}.  Therefore, the limit $\frac{h}{\eta Z_T^\lambda} \in \overline{\mc{Y}}_\tau^0$, implying that $\frac{h}{\eta} \in \overline{\mc{Y}}_\tau^\lambda.$  

Note that $ 0 \leq E \left [ \eta_k Z_T^{\lambda_k} \mc{E}(L^k)_T \ | \ \mc{F}_\tau \right ] =  \eta_k E \left [ Z_T^{\lambda_k} \mc{E}(L^k)_T \ | \ \mc{F}_\tau \right ] \leq \eta_k,$ where the last inequality is due to the supermartinagle property of deflators.  Since $\eta_k \ra \eta$ a.s., it follows that the collection $\left (E \left [\eta_k Z_T^{\lambda_k}\mc{E}(L^k)_T \ | \ \mc{F}_\tau \right ] \right )_{k \geq 1}$ is bounded in probability, as well as the collection $\left\{E[ h_k \ | \ \mc{F}_\tau]\right\}_{k \geq 1}$.  By Lemma \ref{convp}, the collection $\{ V^-(h_k) \}_{k \geq 1}$ is $\mc{F}_\tau$-uniformly integrable.  Applying Lemma \ref{confatou} for the first inequality below, we have 
\[
\begin{split}
v_\tau^\lambda(\eta) &\leq E [ V(h) \ | \ \mc{F}_\tau ] = E \left [ V \left (\liminf_{k \ra \infty} h_k \right ) \ | \ \mc{F}_\tau \right ] \\
&\leq P-\liminf_{k  \ra \infty} E \left [ V \left ( \sum_{j=k}^{J(k)} \alpha_j^k \eta_j Z_T^{\lambda_j} \mc{E}(L^j)_T \right ) \ | \ \mc{F}_\tau \right ] \\
&\leq P- \liminf_{k \ra \infty} \sum_{j=k}^{J(k)} \alpha_j^k E \left [ V \left (\eta_j Z_T^{\lambda_j} \mc{E}(L^j)_T \right ) \ | \ \mc{F}_\tau \right ] 
= P - \liminf_{k \ra \infty} \sum_{j = k}^{J(k)} \alpha_j^k v_\tau^{\lambda_j}(\eta_j).
\end{split}
\]

These calculations imply that 

\[v^\lambda_\tau(\eta)1_{A^c} \leq \left(P - \liminf_{k \ra \infty} \sum_{j = k}^{J(k)} \alpha_j^k v_\tau^{\lambda_j}(\eta_j) \right) 1_{A^c} = P - \liminf_{k \ra \infty} \sum_{j = k}^{J(k)} \alpha_j^k v_\tau^{\lambda_j}(\eta_j) 1_{A^c}.\]

By the standard Borel-Cantelli method, we can pass to another subsequence so that the $P-\liminf$ above is less than the classical $\liminf$ of this subsequence.  Taking expectations and applying Fatou's Lemma, we have
$\mathfrak{v}_\tau^\lambda(\eta) \leq \underset{k \ra \infty}{\liminf} \ \sum_{j = k}^{J(k)} \alpha_j^k \mathfrak{v}_\tau^{\lambda_j} (\eta_j)$ $ \leq \mathfrak{v}_\tau^\lambda(\eta) - \beta(\beta - \epsilon)$, a contradiction.

\end{proof}

\begin{proposition}\label{vcont}  The mapping $(\eta,\lambda) \mapsto v_\tau^\lambda(\eta)$ as defined above is continuous.
\end{proposition}

\begin{proof}  Thanks to Lemma \ref{usc}, it is enough to show that the map is lower semi-continuous.  

 For $n \in \mb{N}, \eta \in L^0_{++}(\mc{F}_\tau)$, and $\lambda \in \Lambda'$, we set
\[ v_\tau^{n,\lambda}(\eta) = v^\lambda_\tau(\eta \vee n^{-1}).\]

Since $V$ is decreasing, we have $v_\tau^{n,\lambda}(\eta) \leq v_\tau^{n+1,\lambda}(\eta)$ for every $n$.  For fixed $\lambda$, Theorem \ref{mainthm1} guarantees the continuity of $v_\tau^\lambda(\cdot)$.  Thus, we know that $v_\tau^{n,\lambda}(\eta) \uparrow v_\tau^\lambda (\eta)$ as $n \ra \infty$.  Thanks to Lemma \ref{auxlsc}, the mappings $(\eta,\lambda) \mapsto v^{n,\lambda}(\eta)$ are lower semi-continuous.  Since $v^{n,\lambda}(\eta) \uparrow v^\lambda(\eta)$, it now follows by Lemma \ref{infsc} that $(\eta,\lambda) \mapsto v^\lambda(\eta)$ is lower semi-continuous, and hence continuous.
\end{proof}

\section{$\lambda$-continuity of $(v_\tau^\lambda)'$}\label{sec:contode}

From the economic motivation of the problem, our initial hypotheses involve appropriate convergence of the markets, as well as convergence of initial wealth.  We want to prove some convergence on the dual side, and then bring things back to the primal.  However, we need a way to obtain convergence of dual ``initial wealths" from convergence of (primal) initial wealths.  Given the relationship $\xi = -(v_\tau^\lambda)'(\eta)$, this implies that we need stability of derivatives with respect to $\lambda$.  This section is occupied with establishing such continuity, which is a classical result in traditional convex analysis, but here requires some additional effort.  The basic strategy is to prove uniform convergence of $v_\tau^\lambda(\eta)$ when $\eta$ ranges over sets with a suitable analog of compactness. Naturally, uniform convergence then leads to the convergence of derivatives.

\subsection{Convex Compactness}\label{sec:ccom}

Given $K \subset L^0_+$, we define the $\mc{F}_\tau$-convex hull of $K$, denoted $\conv_{\mc{F}_\tau}(K)$, to be the set of all finite $\mc{F}_\tau$-convex combinations of elements in $K$.

\begin{definition}  Let $K \subset L^0_+$.  We say that $K$ is $\mc{F}_\tau$-convexly compact if
\begin{enumerate}
\item $K$ is $\mc{F}_\tau$-convex and closed with respect to convergence in probability.
\item For any sequence $(k_n)$ in $K$, there exist $h_n \in \conv_{\mc{F}_\tau}(k_n,\ldots)$ such that $h_n \ra h \in K$ a.s.
\end{enumerate}
\end{definition}

In the rest of this section, $K$ will denote an arbitrary $\mc{F}_\tau$-convexly compact set.

\begin{lemma}\label{awayzero}  Let $K \subset L^0_{++}(\mc{F}_\tau)$.  Then for any $\alpha > 0$, the random variable $X^* = X^*(\alpha) = \underset{X \in K}{\es} X^{-\alpha}$ is finite.
\end{lemma}

\begin{proof}  By standard arguments, the collection $\{X^{-\alpha}\ : \ X \in K\}$ is upwards directed.  Hence, take $X_n \in K$ such that $X_n^{-\alpha} \uparrow X^*$.  Suppose that there exists a set $C$ with $P(C)> 0$ on which $X^*$ is infinitely large.  Then $X_n \ra (X^*)^{-\frac{1}{\alpha}}$, but this last random variable is not contained in $L^0_{++}$.  Since $K$ is closed, it cannot be that it is contained in $L^0_{++}(\mc{F}_\tau)$, a contradiction.
\end{proof}

\begin{lemma}\label{above}  Let $K \subset L^0_{++}(\mc{F}_\tau)$, and let $\lambda_n \ra \lambda$ appropriately.  Let $v_\tau^*(\eta) = \underset{n}{\es} v_\tau^n(\eta)$.  Then $\underset{\eta \in K}{\es} v_\tau^*(\eta) < \infty$.  
\end{lemma}

\begin{proof}  Recall by definition that $v_\tau^n(\eta) \leq E [ V(\eta Z_T^{\lambda_n}) \ | \ \mc{F}_\tau]$.  By the assumption on the asymptotic elasticity of $U$, we have $V(\eta Z_T^{\lambda_n}) \leq C \eta^{-\gamma}V(Z_T^{\lambda_n})+D$, where $\gamma > 0$ is a constant derived from the asymptotic elasticity of $U$ and $C,D$ are constants independent of the choice of $\eta$ or $n$.   This estimate implies that $v_\tau^*(\eta) \leq C \eta^{-\gamma} \underset{n}{\es} E [V(Z_T^{\lambda_n}) \ | \ \mc{F}_\tau]+D$.  By the $V$-compactness hypothesis, $\underset{n}{\es} E [V(Z_T^{\lambda_n}) \ | \ \mc{F}_\tau] \triangleq Z_K < \infty$.  We conclude that $\underset{\eta \in K}{\es} v_\tau^*(\eta) \leq C Z_K \underset{\eta \in K}{\es} \eta^{-\gamma}+D$.  By Lemma \ref{awayzero}, this last quantity is finite almost surely.
\end{proof}

\begin{lemma}\label{convexmax}  Let $K \subset L^0_{++}(\mc{F}_\tau)$.  Then $k^* = \underset{k \in K}{\es} k$ is finite.
\end{lemma}
\begin{proof}  By $\mc{F}_\tau$-convexity, the set $\{k \ : \ k \in K\}$ is upwards directed.  Hence, we can take a sequence $k_n \uparrow k^*$.  Suppose that $k^*$ were not finite.  By passing to $\mc{F}_\tau$-forward convex combinations of the of the $k_n$, we can assume that they converge almost surely to some random variable $k \in K$.  But almost sure convergence is preserved under convex combinations, and so $k = k^*$, and this is a contradiction, since elements of $K$ are real-valued random variables.
\end{proof}

\begin{lemma}\label{convexmin}  Let $K \subset L^0_{++}(\mc{F}_\tau)$, and let $(v_\tau)_*(\eta) = \underset{n}{\ei} v^n_\tau(\eta)$.  Then $(v_\tau)_{**} \triangleq \underset{\eta \in K}{\ei} (v_\tau)_*(\eta) > -\infty$.  
\end{lemma}

\begin{proof}  Note that for fixed $\eta$, $(v_\tau)_*(\eta) > -\infty$, because $(v_\tau^n(\eta))_{n \geq 1}$ is cauchy for all $\eta$.  Fix some $\eta_0 \in K$, and let $\beta_1 = (v_\tau)_*(\eta_0)$.  Let $\epsilon = \epsilon(\omega) \in L^0_{++}(\mc{F}_\tau)$ be such that for any $f \in L^0(\mc{F}_\tau)$ with $|f| \leq \epsilon$, we have $\eta_0 + f \in L^0_{++}(\mc{F}_\tau)$. The collection $\mc{X} = \{\eta_0 + f \ : \ |f| \leq \epsilon\}$ is an $\mc{F}_\tau$-convexly compact set, using Komlos's Lemma.  Hence, by Lemma \ref{above}, let $\beta_2$ be an upper bound for $(v_\tau)^*$ on this set.

Let $\eta$ be an arbitrary element of $K$, distinct from $\eta_0$.  We are trying to show that $(v_\tau)_*$ is uniformly bounded from below on the set $K$.  Since each $v_\tau$ has the local property, so too does $(v_\tau)_*$.  In particular, this means that on the set $\{\eta_0 = \eta\}$, $(v_\tau)_*(\eta_0)$ agrees with $(v_\tau)_*(\eta)$.  Thus, we may assume without loss of generality that $\eta_0 \neq \eta$.

Let $$z = \eta_0 + \frac{\epsilon}{|\eta_0 - \eta|}(\eta_0 - \eta), \quad \text{and}\quad \lambda = \frac{\epsilon}{\epsilon + |\eta_0 - \eta|}.$$  A simple calculation shows that $\eta_0 = (1-\lambda)z + \lambda \eta$.  Furthermore, the $\mc{F}_\tau$-measurable random variable $\lambda$ lives between $0$ and $1$, and $z \in \mc{X}$.  Thus, we have, for any $n$, $$\beta_1 \leq v_\tau^n(\eta_0) \leq (1-\lambda)v_\tau^n(z) + \lambda v_\tau^n(\eta) \leq (1-\lambda)\beta_2  + \lambda v_\tau^n(\eta) \leq \beta_2 + \lambda v_\tau^n(\eta).$$  Consequently, we obtain $v_\tau^n(\eta) \geq \frac{\beta_1 - \beta_2}{\lambda} = \frac{(\epsilon + |\eta_0 - \eta|)(\beta_1 - \beta_2)}{\epsilon}$.  Taking infimums over all $n$, we obtain $$(v_\tau)_*(\eta) \geq \frac{(\epsilon + |\eta_0 - \eta|)(\beta_1 - \beta_2)}{\epsilon} \geq \frac{(\epsilon + \eta_0 + \eta)(\beta_1 - \beta_2)}{\epsilon},$$ noting that the quantity $\beta_1 - \beta_2$ is negative.  Hence, by Lemma \ref{convexmax}, it follows that $(v_\tau)_*$ is uniformly bounded from below on $K$.
\end{proof}

\begin{remark}  Compare Lemma \ref{convexmin} to Theorem 10.6 in \cite{MR1451876}.  The ``neighborhood'' we use does not contain any open sets.  Furthermore, the domain for $v_\tau$, which is $L^0_{++}(\mc{F}_\tau)$, is not an open set in $L^0(\mc{F}_\tau)$.
\end{remark}

\begin{lemma}\label{Lips}  For all $n$, and for all $x,y \in K$, we have $|v_\tau^n(x) - v_\tau^n(y)| \leq \alpha |x -y|$, where $\alpha \in L^0_{++}(\mc{F}_\tau)$ does not depend on $n$.
\end{lemma}

\begin{proof}  In Lemma \ref{awayzero}, the proven result is equivalent to the fact that $\underset{k \in K}{\ei} k > 0$.  Thus, there exists a random variable $\epsilon \in L^0_{++}(\mc{F}_\tau)$ such that the set $\mc{X} = \{k + f \ : k \in K, \ |f| \leq \epsilon\}$ is contained in $L^0_{++}(\mc{F}_\tau)$.  It is also clear that $\mc{X}$ is $\mc{F}_\tau$-convexly compact.  According to the Lemmas \ref{above} and \ref{convexmin}, let $\alpha_1,\alpha_2 \in L^0(\mc{F}_\tau)$ be, respectively, upper and lower bounds for all of the $v_\tau^n$ on $\mc{X}$.

Let $x$ and $y$ be two distinct points in $K$, and let $z = y + \frac{\epsilon}{|y - x|}(y-x)$.  Then $z \in \mc{X}$, and $$y = (1-\lambda)x + \lambda z,$$ for $\lambda = \frac{|y-x|}{\epsilon + |y-x|}$.  Let $n$ be arbitrary but fixed.  By the $\mc{F}_\tau$-convexity of $v_\tau^n$, we obtain $$v_\tau^n(y) \leq (1-\lambda)v_\tau^n(x) + \lambda v_\tau^n(z) = v_\tau^n(x) + \lambda(v_\tau^n(z) - v_\tau^n(x)),$$ and consequently, $v_\tau^n(y) - v_\tau^n(x) \leq \lambda(\alpha_2 - \alpha_1) \leq \alpha |y -x|$, for $\alpha = \frac{\alpha_2 - \alpha_1}{\epsilon}$.  By switching the places of $x$ and $y$, we can obtain $|v_\tau^n(y) - v_\tau^n(x)| \leq \lambda(\alpha_2 - \alpha_1) \leq \alpha |y - x|$.  
\end{proof}

\begin{corollary}\label{equicont}  The $v_\tau^n$ are equi-uniformly-continuous on $K$ as above.
\end{corollary}

\subsection{Another Kind of Compactness}\label{sec:akc}

Note that if $K \subset L^0_{++}(\mc{F}_\tau)$ is $\mc{F}_\tau$-convexly compact, then $K$ is closed under pointwise minimization (and maximization).  Indeed, let $k_1,k_2 \in K$.  Let $A = \{k_1 \leq k_2\} \in \mc{F}_\tau$.  Then $k_1 \wedge k_2 = 1_A k_1 + 1_{A^c} k_2 \in K$.

In this section, we want to prove the following result:

\begin{proposition}\label{subcompact}  Let $\{k_1,k_2,\ldots \}$ be a sequence in $K$.  There exist $y^n_i \in K$ such that $y^n_i \leq k_i$ and $\mc{F}_\tau$-measurable partitions $\pi^n = \{A^n_n,A^n_{n+1},\ldots,A^n_{J(n)}\}$ such that for $$f_n \triangleq \sum_{j=n}^{J(n)} 1_{A_j^n} y^n_j,$$ there exists a random variable $f \in K$ such that $f_n \ra f$ almost surely.
\end{proposition}

To set notation, we call such an $f_n$ above an $\mc{F}_\tau$-partition subcombination of $\{k_1,\ldots\}$ in $K$. 

\noindent \textbf{Proof of Proposition~\ref{subcompact}}

The proof is inspired by the proof of Komlos's Lemma in Section $9.8$ of \cite{MR2200584}.  For a positive integer $n$, define $$K_n^\Pi = \left \{\sum_{j=n}^N 1_{A_j} y_j \ : \ N \in \mb{N},\ \{A_n,\ldots,A_N\} \text{ is an } \mc{F}_\tau-\text{partition of } \Omega, y_j \in K,y_j \leq k_j \right \}.$$ 

Since $K$ is $\mc{F}_\tau$-convex, it follows that $K_n^\Pi \subset K$ for each $n$.  We know that $\underset{k \in K}{\es} v_\tau^\lambda(k) < \infty$, by Lemma \ref{above}.  It thus follows that $g_n \triangleq \underset{k \in K_n^\Pi}{\es} v_\tau^\lambda(k) < \infty$.  Furthermore, the sequence $(g_n)$ is decreasing since $K_n^\Pi \supset K_{n+1}^\Pi$ for all $n$.  Since $\underset{k \in K}{\ei} v_\tau^\lambda(k) > -\infty$ by Lemma \ref{convexmin} (this result is rather stronger than what we need), it follows that $ \underset{n \ra \infty}{\lim} g_n \triangleq g \in L^0(\mc{F}_\tau)$.

Note that each $K_n^\Pi$ inherits $\mc{F}_\tau$-convexity from $K$.  This and the locality of $v_\tau^\lambda$ imply that the sets $\{v_\tau^\lambda(k) \ : \ k \in K_n^\Pi\}$ are upwards directed.  Hence, for each $n$, we take a sequence $(g_n^j)$ in $v_\tau^\lambda(K_n^\Pi)$ such that $g_n^j \uparrow g_n$ almost surely.  Extract a ``diagonal'' $g_n^{j_n} \in K_n^\Pi$ so that $g_n^{j_n} \ra g$ in probability.  Here, $(j_n)$ is just some increasing subsequence of the natural numbers.

Given that $g_n^{j_n} \in v_\tau^\lambda(K_n^\Pi)$, take $x_n \in K_n^\Pi$ such that $v_\tau^\lambda(x_n) = g_n^{j_n}$.  We claim that the sequence $(x_n)$ is Cauchy in probability.  Recall the function $V:\mb{R}_+ \ra \mb{R}$ from which $v_\tau^\lambda$ is defined.  Note that $V$ is strictly decreasing.  This means that for $\alpha > 0$, there is a $\beta > 0$ such that for $x,y \in \mb{R}_+$, if $x-y > \alpha$ and $x \wedge y \leq \alpha^{-1}$, we have $V(y) - V(x) > \beta$.  For fixed $n \geq m$ and $\alpha > 0$, let $A^+ = A^+(n,m,\alpha) = \{x_n - x_m \geq \alpha,x_n \wedge x_m \leq \alpha^{-1}\}$, and let $A^- = A^-(n,m,\alpha) = \{x_n - x_m \leq -\alpha,x_n \wedge x_m \leq \alpha^{-1}\}$.

Thus, we have $v_\tau^\lambda(x_n \wedge x_m) \geq v_\tau^\lambda(x_n) + \beta 1_{A^-}$ and $v_\tau^\lambda(x_n \wedge x_m) \geq v_\tau^\lambda(x_m) + \beta 1_{A^-}$, using the locality of $v_\tau^\lambda$.  Combining these two inequalities, we obtain $v_\tau^\lambda(x_n \wedge x_m) \geq \frac{1}{2}(v_\tau^\lambda(x_n) + v_\tau^\lambda(x_m)) + \beta 1_A$, where $A$ is defined as $A^+ \cup A^-$.  This implies that $$\beta 1_A \leq v_\tau^\lambda(x_n \wedge x_m) - \frac{1}{2}(v_\tau^\lambda(x_n) + v(x_m)).$$  One verifies directly that $x_n \wedge x_m \in K_n^\Pi$, and so the quantity above is less than or equal to $g_n - \frac{1}{2}(v_\tau^\lambda(x_n) + v_\tau^\lambda(x_m)) = g_n - \frac{1}{2}(g_n^{j_n} + g_m^{j_m})$.  Letting $n$ and $m$ go to infinity, it follows that for any $\alpha>0$, $P(A(n,m,\alpha)) \ra 0$ as $n$ and $m$ go to infinity.

Given that $K$ is itself bounded in probability by Lemma \ref{convexmax}, it follows that the $x_n$ are also bounded in probability, and from this fact and the above paragraph we deduce that the collection $(x_n)$ is indeed Cauchy in probability.  Thus, there exists $x \in K$ which is the limit in probability of the $(x_n)$.  We obtain almost sure convergence by passing to a subsequence.
\hfill $\square$

\subsection{Nets and Convergence of Derivatives}

When dealing with compact sets, $\epsilon$-nets are useful tools.  In the less restrictive setting of convexly compact sets, we have to use different kinds of nets.

\begin{definition}  Let $K \subset L^0(\mc{F}_\tau)$, let $\mc{K} = \{k_1,\ldots,k_n\}$ be a finite collection of points in $K$, and let $r>0$.  We say that $\mc{K}$ is a \textbf{$\mc{F}_\tau$-convex $r$-net} of $K$ if for any $x \in K$, there exists $y \in \conv_{\mc{F}_\tau}(\mc{K})$ such that $d(x,y) \leq r$, where $d$ is a distance function compatible with convergence in probability.
\end{definition}

\begin{lemma}\label{convexnet}  Let $K \subset L^0(\mc{F}_\tau)$.  Then for any $r>0$, there is a finite $\mc{F}_\tau$-convex $r$-net of $K$.
\end{lemma}

\begin{proof}  Suppose to the contrary that such a net does not exist.  This means that one can inductively construct a sequence $(k_n)$ in $K$ such that for all $n$, $d(k_{n+1},y) > r$ for all $y_n \in \conv_{\mc{F}_\tau}(k_1,\ldots,k_n)$.  Since $K$ is $\mc{F}_\tau$-convexly compact, take forward $\mc{F}_\tau$-convex combinations $f_n$ of the $k_n$ converging almost surely to $f$.  This of course implies that as $n \ra \infty$, $d(f_{n+1},f_n) \ra 0$, a contradiction.  
\end{proof}

\begin{definition}  Given a set $K \subset L^0(\mc{F}_\tau)$ and $K' \subset K$, we define the \textbf{$\mc{F}_\tau$-partition sub-convex hull} of $K'$ in $K$, denoted by $\conv_{\mc{F}_\tau}^\Pi(K',K)$, to the be the set of $\mc{F}_\tau$-partition subcombinations of elements of $K'$ in $K$.  More precisely, it consists of elements of the form 
\[ \sum_{n = 1}^N 1_{A_n} y_n,\]
where $N$ is an integer, $\{A_1,\ldots,A_n\}$ is an $\mc{F}_\tau$-partition of $\Omega$, and $y_n \in K$ such that $y_n \leq k'_n$ for some $k'_n \in K'$.
\end{definition}

\begin{definition} Let $K \subset L^0$, and let $\mc{K} = \{k_1,\ldots,k_n\}$ be a finite collection of points in $K$, and let $r>0$.  We say that $\mc{K}$ is a \textbf{$\mc{F}_\tau$-partition sub-convex $r$-net} of $K$ if for any $x \in K$, there exists $y \in \conv_{\mc{F}_\tau}^\Pi(\mc{K},K)$ such that $d(x,y) \leq r$.
\end{definition}

\begin{lemma}  Let $K \subset L^0_{++}(\mc{F}_\tau)$.  Then for any $r>0$, there is a finite $\mc{F}_\tau$-partition sub-convex $r$-net of $K$.
\end{lemma}

\begin{proof}  The proof is similar to the proof of Lemma~\ref{convexnet}: One needs to use the concept introduced in Proposition \ref{subcompact} in place of $\mc{F}_\tau$-forward convex combinations.
\end{proof}

\begin{proposition}\label{unifder}  Let $K \subset L^0_{++}(\mc{F}_\tau)$ be $\mc{F}_\tau$-convexly compact. Then $v_\tau^n \ra v_\tau^\lambda$ uniformly on $K$.  
\end{proposition}

\begin{remark}  In the classical proof of this result in $\mb{R}^N$ with compact sets, one simply defines a standard $\epsilon$-net and goes from uniform convergence on that net to uniform convergence everywhere by virtue of an equicontinuity property.  In this proof, the second part of the above argument is the same, but the first needs to be altered, because we need to ensure uniform convergence on a kind of convex hull, which has infinitely many points.  Each of the nets described above gives one-sided inequalities for convergence in its respective ``convex hull''; by combining the two nets, we can get uniform convergence.
\end{remark}

\begin{proof}  
We have already shown that $v_\tau^n \ra v_\tau^\lambda$ pointwise; see Proposition~\ref{vcont}. In the rest of the proof we will upgrade the pointwise convergence to uniform convergence over $K$. 
Let $\epsilon>0$ be arbitrary but fixed.  For this $\epsilon$, let $\delta>0$ be the the equi-continuity constant whose existence is implied by Corollary \ref{equicont}; that is, for any $a_1,a_2 \in K$, $d(a_1,a_2)<\delta$ implies that $$\max\{d(v_\tau^\lambda(a_1),v_\tau^\lambda(a_2)),d(v_\tau^n(a_1),v_\tau^n(a_2)) \}< \epsilon.$$  By the above two lemmas, we know that we can construct a finite $\mc{F}_\tau$-convex $\delta$-net of $K$ and a finite $\mc{F}_\tau$-partiton sub-convex $\delta$-net of $K$.  By just taking the union of these sets, we will assume that there is a finite set $\mc{K} = \{k_1,\ldots,k_s\}$ that simultaneously is both of these types of nets.

Given that $\mc{K}$ is finite, we know that for $n$ sufficiently large, $d(v_\tau^n(k_i),v_\tau^\lambda(k_i)) < \frac{\epsilon}{s}$ for all $i$.  Let $x \in \conv_{\mc{F}_\tau}(\mc{K})$, and write $$ x = \sum_{i=1}^s g_i k_i,\text{ where } \sum_{i=1}^s g_i \equiv 1$$ and each $g_i$ is $\mc{F}_\tau$-measurable and between zero and one.  By the $\mc{F}_\tau$-convexity of $v_\tau^n$, we have \begin{equation}\label{up} v_\tau^n(x) \leq \sum_{i=1}^s g_i v_\tau^n(k_i).\end{equation}  Let $y \in \conv_{\mc{F}_\tau}^\Pi(\mc{K},K)$ be such that $d(x,y) < \delta$.  This means that 
\begin{equation}\label{0} \max\{d(v_\tau^\lambda(x),v_\tau^\lambda(y)),d(v_\tau^n(x),v_\tau^n(y))\}<\epsilon. \end{equation}  
Write $y = \sum_{i=1}^s 1_{A_i} y_i,$ where the $A_i$ and $y_i$ satisfy the usual properties.  By locality, \begin{equation}\label{do} v_\tau^n(y) = \sum_{i=1}^s 1_{A_i} v_\tau^n(y_i) \geq \sum_{i=1}^s 1_{A_i} v_\tau^n(k_i).\end{equation}

We can establish the trivial inequalities \begin{equation}\label{3}d \left ( \sum_{i=1}^s g_i v_\tau^n(k_i), \sum_{i=1}^s g_i v_\tau^\lambda(k_i) \right ) < s \frac{\epsilon}{s} = \epsilon \end{equation} and \begin{equation}\label{4} d \left ( \sum_{i=1}^s 1_{A_i} v_\tau^n(k_i), \sum_{i=1}^s 1_{A_i} v_\tau^\lambda(k_i) \right ) < s \frac{\epsilon}{s} = \epsilon.\end{equation}  Combining $\eqref{up}-\eqref{4}$, we can conclude that $d(v_\tau^\lambda(x),v_\tau^n(x)) < 5 \epsilon$.

  Now we extend the result to all of $K$.  For an arbitrary $z \in K$, choose $x \in \conv_{\mc{F}_\tau}(\mc{K})$ such that $d(x,z) < \delta$.  Then the equicontinuity implies that $d(v_\tau^n(z),v_\tau^\lambda(z))<6\epsilon$ for $n$ sufficiently large.
\end{proof}

\begin{proposition}\label{derivconv}  For all $\eta \in L^0_{++}(\mc{F}_\tau),b \in L^1(\mc{F}_\tau)$, $(v_\tau^n)'(\eta;b) \ra (v_\tau^\lambda)'(\eta;b)$ .
\end{proposition}

\begin{proof}  Recall that a sequence of random variables converges in probability if and only if, for any $\epsilon>0$, there exists a set of measure greater than $1-\epsilon$ such that the random variables when restricted to this set converge in probability.  Thus, by definition, it suffices to prove the result for $b \in L^\infty(\mc{F}_\tau)$ and $\eta$ bounded from below in $L^0_{++}(\mc{F}_\tau)$.  There exists $\epsilon > 0$ such that $\mc{X} = \{\eta + f \ : \ |f| \leq \epsilon\}$ is contained in $L^0_{++}(\mc{F}_\tau)$.  Furthermore, the set $\mc{X}$ is $\mc{F}_\tau$-convexly compact.  For $t \in \mb{R}$ with $|t|$ small enough, we know that $\eta + tb \in \mc{X}$.  The proposition is now a straightforward consequence of the uniform convergence established in Proposition \ref{unifder}.
\end{proof}

\begin{proposition}\label{derivconv2}  For $\eta \in L^0_{++}(\mc{F}_\tau)$, $(v_\tau^n)'(\eta) \ra (v_\tau^\lambda)'(\eta)$.
\end{proposition}

\begin{proof}  
This is a direct consequence of Proposition~\ref{derivconv} since
$v_\tau^\lambda(\eta;b) = b(v_\tau^\lambda)'(\eta)$.
\end{proof}

We may without any extra cost strengthen the above pointwise result into one of uniform convergence, using again Proposition \ref{unifder}.

\begin{proposition}\label{dualderunif}Let $K \subset L^0_{++}(\mc{F}_\tau)$ be $\mc{F}_\tau$-convexly compact.  Then $(v_\tau^n)'(\eta) \ra (v_\tau^\lambda)'(\eta)$, uniformly over all $\eta \in K$.
\end{proposition}

The results on the dual side are also applicable to the primal value functions.  We simply apply them to the $\mc{F}_\tau$-convex function $-u_\tau^\lambda$.

\begin{proposition}\label{primalunif} Let $K\subset L^0_{++}(\mc{F}_\tau)$ be $\mc{F}_\tau$-convexly compact.  Then $u_\tau^n(\xi) \ra u_\tau^\lambda(\xi)$ and $(u_\tau^n)'(\xi) \ra (u_\tau^\lambda)'(\xi)$ uniformly over all $\xi \in K$.

\end{proposition}

\section{From the Dual to the Primal}\label{sec:fdtp}

\subsection{Continuity of Value Functions}

\begin{proposition}\label{contmaps}  The following four maps are continuous on $L^0_{++}(\mc{F}_\tau) \times \Lambda'$:
\[ (\eta,\lambda) \mapsto v_\tau^\lambda(\eta), \ (\eta,\lambda) \mapsto (v_\tau^\lambda)'(\eta), \ (\xi,\lambda) \mapsto u_\tau^\lambda(\xi), \text{ and } (\xi,\lambda) \mapsto (u_\tau^\lambda)'(\xi).\]
\end{proposition}

\begin{proof}  We know the continuity of the first two functions with respect to $\lambda$ by Proposition \ref{vcont} and Proposition \ref{derivconv2}.

Obtaining continuity of the third and fourth functions with respect to $\lambda$ is done as in Proposition $3.9$ of \cite{MR2438002}.  Choose $\xi \in L^0_{++}(\mc{F}_\tau)$ and $\epsilon > 0$, and define $\eta(\epsilon) = (u_\tau^\lambda)'(\xi) + \epsilon$.  Since $(v_\tau^\lambda)'(\cdot)$ is strictly increasing 
\[ \underset{n \ra \infty}{\lim} (v_\tau^n)'(\eta(\epsilon)) = (v_\tau^\lambda)'(\eta(\epsilon)) = (v_\tau^\lambda)'((u_\tau^\lambda)'(\xi) + \epsilon) > (v_\tau^\lambda)'((u_\tau^\lambda)'(\xi)) = - \xi,\]
the last inequality following from the strict increase of $(v_\tau^\lambda)'$, and the last equality following from the conjugacy of $(v_\tau^\lambda)'$ and $(u_\tau^\lambda)'$.  Consequently, for large $n$, we have $-(v_\tau^n)'(\eta(\epsilon)) < \xi$.  Since $(u_\tau^n)'$ is strictly decreasing for each $n \in \mb{N}$, we have
\[ (u_\tau^\lambda)'(\xi) + \epsilon = \eta(\epsilon) = (u_\tau^n)'(-(v_\tau^n)'(\eta(\epsilon))) > (u_\tau^n)'(\xi),\]
for large $n$, implying that $\underset{n \ra \infty}{\limsup} (u_\tau^n)'(\xi) \leq (u_\tau^\lambda)'(\xi)$.  The other inequality, that $\underset{n \ra \infty}{\liminf} (u_\tau^n) \geq (u_\tau^\lambda)'(\xi)$, may be proven similarly, using $-\epsilon$ in place of $\epsilon$.  This gives continuity of $(u_\tau^\lambda)'$ with respect to $\lambda$.

We establish the joint continuity for $v_\tau^\lambda$, using Proposition \ref{unifder}.  Joint continuity for the other three mappings will be established in the same way, using Proposition \ref{dualderunif} and Proposition \ref{primalunif}.  Let $(\eta_n,\lambda_n) \ra (\eta,\lambda)$.  We want to show that $v_\tau^n(\eta_n) \ra v_\tau^\lambda(\eta)$ in probability.  Hence, by passing to subsequences, it suffices to assume that $\eta_n \ra \eta$ almost surely.  In this case, the set $\{\eta_n\}_{n} \cup \{\eta\}$ is $\mc{F}_\tau$-convexly compact, because all forward $\mc{F}_\tau$-convex combinations converge to $\eta$.  Hence, we apply Proposition \ref{unifder} to get $v_\tau^n(\eta^*) \ra v_\tau^\lambda(\eta^*)$ uniformly over all $\eta^* \in \{\eta_n\}_{n} \cup \{\eta\}$, and the result now follows.
\end{proof}

\subsection{Continuity of Terminal and Intermediate Wealths}

We now give the proof of the second main theorem:

\noindent \textbf{Proof of Theorem \ref{mainthm2}}

Continuity of the first map is established in Lemma \ref{contmaps}.  For the second part,
by Theorem~\ref{mainthm1}, the optimal terminal wealth $\hat{X}_T^{\xi,\lambda} \in \overline{\mc{X}}^\lambda_\tau(\xi)$ admits the representation, with $\eta = (u_\tau^\lambda)'(\xi)$, of $U'(\hat{X}_T^{\xi,\lambda}) = \eta\hat{Y}_T^{\eta,\lambda}$, where $\hat{Y}_T^{\eta,\lambda} \in \overline{\mc{Y}}^\lambda_\tau$ solves the dual minimization problem.  Thanks to the continuity of the mappings $(\xi,\lambda) \mapsto (u_\tau^\lambda)'(\xi)$ and $x \ra (U')^{-1}(x)$, it suffices to show that $(\eta,\lambda) \mapsto \eta \hat{Y}_T^{\eta,\lambda}$ is continuous.  The proof, at this point, is done identically to Lemma $3.10$ of \cite{MR2438002}, the main ingredients being the strict convexity of $V$ and continuity of the dual value function.  We refer the reader to \cite{MR2438002} for details.
\hfill $\square$

We now lay the groundwork to prove Corollary \ref{intcont}.

\begin{lemma}\label{valmart}  Let $\hat{X}_T^{x,\lambda}$ be the optimal terminal wealth for the market $\lambda$, starting at wealth $x$.  Then the process $U_t = u_t^\lambda(\hat{X}_t^{x,\lambda})$ is a martingale.
\end{lemma} 

\begin{proof}  By definition, $u_0^\lambda(x) = E [ U(\hat{X}_T^{x,\lambda})]$.  We argue that $u_t^\lambda(\hat{X}_t^{x,\lambda}) = E [ U(\hat{X}_T^{x,\lambda}) \ | \ \mc{F}_t]$, which is sufficent to prove that $U$ is a martingale.  Let $t \in [0,T]$.  By Theorem \ref{mainthm1}, there exists $X \in \mc{X}^\lambda_t$ such that
\[ u_t^\lambda(\hat{X}_t^{x,\lambda}) = E [ U(\hat{X}_t^{x,\lambda}  X_T) \ | \ \mc{F}_t].\]
Suppose that 
\[ E [ U(\hat{X}_T^{x,\lambda}) \ | \ \mc{F}_t] \neq E [ U(\hat{X}_t^{x,\lambda} X_T) \ | \ \mc{F}_t].\]
Let $A \in \mc{F}_t$ be the set on which the first quantity is larger, and consider $X'_T \triangleq 1_A \frac{\hat{X}_T^{x,\lambda}}{\hat{X}_t^{x,\lambda}} + 1_{A^c} X \in \overline{\mc{X}}^\lambda_t$.  One then calculates that $ E [ U(\hat{X}_t^{x,\lambda} X') ] > E [U(\hat{X}_T^{x,\lambda})]$.  This, however, contradicts the fact that $\hat{X}_T^{x,\lambda}$ was optimal at time zero, since $\hat{X}_t^{x,\lambda}X' \in \overline{\mc{X}}^\lambda$.
\end{proof}

\begin{lemma} \label{lem:ucp-u}  Suppose that $(x_n,\lambda_n) \ra (x,\lambda)$ appropriately.  Then $u^n(\hat{X}^{x_n,n}) \ra u^\lambda(\hat{X}^{x,\lambda})$ in the ucp (uniform convergence in probability) sense.
\end{lemma}

\begin{proof}  By Lemma \ref{valmart}, the processes $u^n(\hat{X}^{x_n,n})$ and $u^\lambda(\hat{X}^{x,\lambda})$ are all martingales.  By Lemma \ref{terminalconv}, $u_T^n(\hat{X}_T^{x_n,n}) = U(\hat{X}_T^{x_n,n}) \overset{L^1}{\ra} U(\hat{X}_T^{x,\lambda}) = u_T^\lambda(\hat{X}_T^{x,\lambda})$.  Applying the weak form of Doob's $L^p$ inequality for $p=1$, we deduce ucp convergence.
\end{proof}

Note that for the problem here, we only need pointwise convergence at each stopping time $\tau$.  The stronger result given above would be appropriate for a setting in which the goal was ucp convergence of the optimal wealth processes.  This is the topic of our next paper.

\noindent \textbf{Proof of Corollary \ref{intcont}}  Let $(x_n,\lambda_n) \ra (x,\lambda)$.  By Lemma~\ref{lem:ucp-u}, for any stopping time $\tau$, we have $u_\tau^n(\hat{X}_\tau^{x_n,n}) \ra u_\tau^\lambda(\hat{X}_\tau^{x,\lambda})$ in probability.  By the reasoning given in Proposition~$\ref{contmaps}$ (namely that continuous strictly increasing functions have continuous inverses), the mapping $(\xi,\lambda) \mapsto (u_\tau^\lambda)^{-1}(\xi)$ is continuous.  We conclude that $\hat{X}_\tau^{x_n,n} \ra \hat{X}_\tau^{x,\lambda}$ in probability.
\hfill $\square$

\appendix
\section{A Conditional Minimax Theorem}\label{sec:cmmt}

In this section, we prove a conditional version of the Minimax Theorem from convex analysis.  It is used to establish the dual relationship between $u^\lambda_\tau$ and $v^\lambda_\tau$.
We let $\mc{G}$ be an arbitrary sub sigma algebra of $\mc{F}$.  In applying the results of this section, it will always be the case that $\mc{G} = \mc{F}_\tau$ for some stopping time $\tau$.

Let $X$ be a bounded, $\sigma(L^\infty,L^1)$-compact, $\mc{G}$-convex subset of $L^\infty_+$, and let $Y$ be a closed, $\mc{G}$-convex, bounded subset of $L^1_+$.  Let $K:X \times Y \ra L^1(\mc{G})$ be the map defined by $K(x,y) = E [U(x) -xy \ | \ \mc{G} ]$.  In this section we prove the following conditional minimax theorem:

\begin{proposition} $\underset{x \in X}{\es} \underset{y \in Y}{\ei} K(x,y) = \underset{y \in Y}{\ei} \underset{x \in X}{\es} K(x,y)$
\end{proposition}

The strategy will be to reduce to the unconditional case.  As usual, the $``\leq''$ direction is trivial, so we prove this first.  For any $y \in Y$ and $x \in X$, we trivially have $\underset{y \in Y}{\ei} K(x,y) \leq\underset{x \in X}{\es} K(x,y)$.  We then take the supremum over all $x$ on the left hand side and the infimum over all $y$ on the right hand side, giving the desired inequality.

As in the unconditional case, the $``\geq''$ inequality is the difficult one.  Suppose for contradiction that $\underset{x \in X}{\es} \underset{y \in Y}{\ei} K(x,y) < \underset{y \in Y}{\ei} \underset{x \in X}{\es} K(x,y)$ on a set of positive measure.  Then there exists an $\epsilon>0$ and a set $A \in \mc{G}$ with $P(A) > 0$ such that $$\underset{x \in X}{\es} \underset{y \in Y}{\ei} K(x,y) \leq  -\epsilon 1_A + \underset{y \in Y}{\ei}  \underset{x \in X}{\es}  K(x,y).$$

Taking expectations of both sides, we have the inequality $$ E \left [\underset{x \in X}{\es} \underset{y \in Y}{\ei}K(x,y) \right ]< -\epsilon P(A) + E \left [ \underset{y \in Y}{\ei}  \underset{x \in X}{\es}K(x,y) \right ].$$

What we want to do now is interchange the expectation with the essential supremum and essential infimums on each side.  We treat the left hand side.  First, it is obvious by the definition of essential supremum that $$E \left [\underset{x \in X}{\es} \underset{y \in Y}{\ei}K(x,y) \right ] \geq \sup_{x \in X} E \left [ \underset{y \in Y}{\ei}K(x,y) \right ].$$

\begin{lemma}  The set $\left \{ K(x,y) \ : \ y \in Y \right \}$ is downwards directed for all $x \in X$.
\end{lemma}

\begin{proof}  Fix $x \in X$.  Let $y_1,y_2 \in Y$, and let $B = \{K(x,y_1) \leq K(x,y_2) \} \in \mc{G}$.  By the $\mc{G}$-convexity of $Y$, $y \triangleq 1_By_1 + 1_{B^c}y_2 \in Y$.  Then it is immediate that $K(x,y) = (1_BK(x,y_1) + 1_{B^c}K(x,y_2)) = K(x,y_1) \wedge K(x,y_2)$. Thus, the set in question is downwards directed.
\end{proof}

By the properties of essential supremum, it follows that for each $x \in X$, there exist $y_n = y_n(x) \in Y$ such that $K(x,y_n(x)) \downarrow \underset{y \in Y}{\ei} K(x,y)$.  So, for fixed $x$, we have, by Monotone Convergence, $$E \left [ \underset{y \in Y}{\ei} K(x,y) \right ] = E \left [ \lim_{n \ra \infty} K(x,y_n) \right ]$$ $$ = \lim_{n \ra \infty} E[K(x,y_n)] \geq \inf_{y \in Y} E[K(x,y)].$$  We can thus conclude that $$\sup_{x \in X} \inf_{y \in Y} E [K(x,y)] \leq E \left [\underset{x \in X}{\es} \underset{y \in Y}{\ei} K(x,y) \right ].$$

We now treat the right hand side in a similar manner.  As before, by the definition of essential supremum, $$E \left [ \underset{y \in Y}{\ei}  \underset{x \in X}{\es}  K(x,y) \right ] \leq \inf_{y \in Y} E \left [ \underset{x \in X}{\es}  K(x,y) \right ].$$

This next lemma is proved identically to the one appearing above.

\begin{lemma}  The set $\left \{ K(x,y) \ : \ x \in X \right \}$ is upwards directed for all $y \in Y$.
\end{lemma}

We continue precisely as before to eventually obtain this inequality: $$E \left [ \underset{y \in Y}{\ei}  \underset{x \in X}{\es}  K(x,y) \right ] \leq \inf_{y \in Y} \sup_{x \in X} E [K(x,y)].$$

We combine the two derived inequalities to obtain 
\begin{equation}\label{eq1}\sup_{x \in X} \inf_{y \in Y} E [K(x,y) ] < -\epsilon P(A) + \inf_{y \in y} \sup_{x \in X } E [K(x,y) ].  
\end{equation}
We are now in a position to obtain a contradiction by applying the unconditional minimax theorem. Following the argument presented on p. 13 of \cite{MR1722287}, we have, using the classical minimax theorem, that
\begin{equation}\label{eq2} \underset{x \in X}{\sup} \ \underset{y \in Y}{\inf} E[U(x) - xy] = \underset{y \in Y}{\inf} \ \underset{x \in X}{\sup} E [ U(x) - xy].
\end{equation}

Combining \eqref{eq1} and \eqref{eq2}, we obtain the desired contradiction.  This establishes the conditional minimax theorem.

\bibliographystyle{siam}
\bibliography{references}

\end{document}